\documentclass[preprint,11pt]{elsarticle}
\usepackage[top=1in,bottom=1.25in,left=1.25in,right=1.25in]{geometry}

\usepackage{amsmath,amssymb,amsthm}
\usepackage{latexsym}
\usepackage{indentfirst}
\usepackage{graphicx}
\usepackage{placeins}
\usepackage{booktabs}
\usepackage{algorithm}
\usepackage{algorithmic}
\usepackage{multirow}
\usepackage{graphicx,color}
\usepackage{diagbox}
\usepackage{newtxtext}
\usepackage{newtxmath}
\usepackage{bm}
\usepackage{dsfont}

\usepackage{subfigure}

\theoremstyle{plain}
\newtheorem{theorem}{Theorem}
\newtheorem{lemma}[theorem]{Lemma}

\theoremstyle{definition}

\theoremstyle{remark}
\newtheorem{remark}{Remark}

\newtheoremstyle{cited}%
  {3pt}
  {3pt}
  {\itshape}
  {}
  {\bfseries}
  {.}
  {.5em}
  {\thmname{#1} \thmnumber{#2} \thmnote{\normalfont#3}}

\theoremstyle{cited}

\newcommand{\eps}{\epsilon}

\newcommand{\norm}[1]{\lVert#1\rVert}
\newcommand{\vertiii}[1]{\left\lvert\kern-0.25ex\left\lvert\kern-0.25ex\left\lvert #1 \right\rvert\kern-0.25ex\right\rvert\kern-0.25ex\right\rvert}

\usepackage{slashed}

\newcommand{\bbm}{\begin{bmatrix}}
\newcommand{\ebm}{\end{bmatrix}}

\newcommand{\Rm}{\mathbb R}

\newcommand{\cS}{\mathcal{S}}

\newcommand{\flatL}{\mathcal{L}_0}
\newcommand{\flatP}{\mathcal{P}}

\newcommand{\cF}{\mathcal{F}}
\newcommand{\cL}{\mathcal{L}}
\newcommand{\cM}{\mathcal{M}}

\newcommand{\cP}{\mathcal{P}}

\newcommand{\cR}{R}

\newcommand{\cT}{\mathcal{T}}

\newcommand{\bxi}{B}

\newcommand{\psikg}{u}
\newcommand{\usn}{u_{s,0}}
\newcommand{\usd}{u_{s,D}}

\newcommand{\Rvar}{\cR_{\text{in}}\;}

\newcommand{\rhovar}{\tau}
\newcommand{\cobmat}{{\rm V}}
\newcommand{\gt}{\chi}
\newcommand{\ldiff}{\Delta}
\newcommand{\mmat}{{\rm M}}

\newcommand{\tmu}{\tilde{\mu}}

\newcommand{\cSflat}{\mathcal{S}_\omega^0}

\newcommand{\cK}{\mathcal{K}}

\newcommand{\dc}{\eta}
\newcommand{\opn}{\mathcal{N}}

\usepackage[utf8]{inputenc}
\usepackage{xcolor}

\journal{Applied and Computational Harmonic Analysis}

\begin{document}

\begin{frontmatter}



\title{Integral formulation of Dirac singular waveguides}

\author[inst1]{Guillaume Bal}

\affiliation[inst1]{organization={Departments of Statistics and Mathematics and CCAM},
            addressline={University of Chicago}, 
            city={Chicago},
            state={IL},
            postcode={60637}, 
            country={USA}}

\author[inst2]{Jeremy Hoskins}

\affiliation[inst2]{organization={Departments of Statistics and CCAM},
            addressline={University of Chicago}, 
            city={Chicago},
            state={IL},
            postcode={60637}, 
            country={USA}}

\author[inst3]{Solomon Quinn}

\affiliation[inst3]{organization={School of Mathematics},
            addressline={University of Minnesota}, 
            city={Minneapolis},
            state={MN},
            postcode={55455}, 
            country={USA}}

\author[inst4]{Manas Rachh}
\affiliation[inst4]{organization={Center for Computational Mathematics},
            addressline={Flatiron Institute}, 
            city={New York},
            state={NY},
            postcode={10010}, 
            country={USA}}

\begin{abstract}
This paper concerns a boundary integral formulation for the two-dimensional massive Dirac equation. The mass term is assumed to jump across a one-dimensional interface, which models a transition between two insulating materials. This jump induces surface waves that propagate outward along the interface but decay exponentially in the transverse direction. After providing a derivation of our integral equation, we prove that it has a unique solution for almost all choices of parameters using holomorphic perturbation theory. We then extend these results to a Dirac equation with two interfaces. Finally, we implement a fast numerical method for solving our boundary integral equations and present several numerical examples of solutions and scattering effects.
\end{abstract}




\end{frontmatter}


\section{Introduction}
The two-dimensional time-harmonic Dirac equation 
$$-i  \begin{pmatrix}0&1 \\ 1& 0 \end{pmatrix}\partial_{x_1}\begin{pmatrix}\psi_1 \\ \psi_2 \end{pmatrix} - i  \begin{pmatrix}0&-i \\ i& 0 \end{pmatrix}\partial_{x_2}\begin{pmatrix}\psi_1 \\ \psi_2 \end{pmatrix}  = E \begin{pmatrix}\psi_1 \\ \psi_2 \end{pmatrix} $$
arises in electronic band theory, where it is commonly used to model materials in which two energy bands intersect conically at a point. It forms an important simplified model in many applications, particularly in graphene, phononics, photonic crystals, and plasmonics \cite{BH,fefferman2012honeycomb,helsing2021dirac}. When a {\it mass} term of the form
$$ m \begin{pmatrix}
    1 & 0 \\ 0 &-1
\end{pmatrix}$$
is added, the energy bands no longer intersect, and for energies $E$ between $m$ and $-m$ the solution decays exponentially away from a source; for this range of energies the material is insulating. The term ``mass'' used here is historical and refers to the ``effective mass'' used in solid states physics to simplify band structures by modeling it as a free particle with that mass. 

If, on the other hand, the mass jumps across an interface from a negative to a positive value, then edge modes can form which propagate along the interface and decay exponentially away from the interface. Remarkably, these modes travel solely in one direction, and this unidirectional transport is robust to perturbations. Physically, this corresponds to two insulators with different ``masses'' being brought together at an interface. The asymmetric transport observed at such interfaces then affords a topological interpretation \cite{BH,Witten}. A physical principle called a bulk-edge correspondence relates a current observable modeling this asymmetric transport to the difference of mass terms in the subdomains $\Omega_1$ and $\Omega_2$. For analyses of the bulk-edge correspondence in various settings, we refer the reader to, e.g., \cite{SB-2000,Elbau,GP,PS,Drouot,bal2022topological, bal2023topological,QB} in the mathematical literature and, e.g., \cite{hatsugai1993edge,Volovik,EG,fukui2012bulk} in the physical literature. For explicit calculations of invariants associated to Dirac equations, see also \cite{2}. A thorough analysis of the Dirac system of equations is presented in \cite{thaller2013dirac}.


This setup can be described by the following set of partial differential equations (PDEs)
\begin{align}\label{eq:Dirac}
\begin{cases}
-i \sigma_1 \partial_x \psi(x) -i \sigma_2 \partial_y \psi (x)+m \sigma_3 \psi(x) - E \psi = f_2(x), \quad   & x \in \Omega_2,\\
-i \sigma_1 \partial_x \psi(x) -i \sigma_2 \partial_y \psi (x)-m \sigma_3 \psi(x) - E \psi = f_1(x), \quad   & x \in \Omega_1,\\
\lim_{y\to x} \psi(y) = \lim_{z \to x} \psi(z) \quad & y \in \Omega_2, z \in \Omega_1, x \in \Gamma,
\end{cases}
\end{align}
where 
\begin{align}\label{eq:Pauli}
\sigma_1 = \begin{pmatrix}
    0 &1 \\ 1& 0
\end{pmatrix}, \quad 
\sigma_2 = \begin{pmatrix}
    0 &-i \\ i & 0 
\end{pmatrix}, \quad {\rm and} \quad \sigma_3 = \begin{pmatrix}
    1 & 0 \\ 0 &-1
\end{pmatrix},
\end{align}
are the Pauli spin matrices. In addition, the domains $\Omega_{1}$ and $\Omega_2$ denote the support of the first and second insulators, respectively, with corresponding masses $-m$ and $m$. Finally, the real scalar quantity $E$ will be referred to as the energy, and the complex vector valued functions $f_1$ and $f_2$ are source terms in $\Omega_1$ and $\Omega_2$, respectively. In the remainder of the paper we will always assume that $|E|^2 < m^2$, that $\overline{\Omega_1 \cup \Omega_2}=\Rm^2$ is the entire plane, and that $\Omega_1$ and $\Omega_2$ meet along an interface $\Gamma=\partial\Omega_1=\partial\Omega_2$. Additionally, we assume that $\Gamma$ is a smooth simple curve which is {\it asymptotically flat} (we refer the reader to Section \ref{subsec:detailed} for a more precise definition). 

In this paper we present a novel boundary integral equation (BIE) formulation of the above PDE system, and prove bounded invertibility under certain natural conditions on the mass $m$ and the interface $\Gamma.$ The approach generalizes the one introduced in \cite{bal2022integral} for a scalar analog of \eqref{eq:Dirac}, and is based on introducing an auxiliary variable involving the application of a certain one-dimensional integral operator. In addition to its favorable analytical properties, this BIE can be easily solved via the combination of high order discretization methods and fast algorithms. 

Surface waves and surface-localized motion are also present in a number of other contexts, particularly in the study of electromagnetic properties of systems in which the ratio of permittivities approach the negative real axis, see \cite{raether1988surface,tzarouchis2018light} and the references therein, as well as \cite{helsing2018spectra,helsing2020extended,helsing2021dirac} for numerical methods. Additionally, the technique employed here is similar to other surface-wave preconditioners, referred to as on-surface radiation conditions, which have been applied to solving high-frequency scattering problems in  electromagnetics, acoustics, and elasticity, see~\cite{darbas2013combining,antoine2008advances,antoine2005alternative,antoine2006improved,chaillat2015approximate,chaillat2014new,kriegsmann1987new} and the references therein. 
These on-surface radiation conditions are generally used to improve the performance of iterative solvers in the high-frequency regime and complicated geometries, as opposed to this work wherein the preconditioner generates surface waves intrinsic to the governing equations.

For Dirac equations with smooth coefficients, volume integral equation-based approaches have also been developed \cite{bal2023asymmetric}, particularly in the evaluation of spectral properties of Dirac operators with volumetric perturbations. In the time domain, after taking a suitable inverse Fourier transform, the propagation of wavepackets along interfaces in Dirac models has been extensively studied, see, e.g., \cite{bal2023edge,bal2023magnetic,drouot2022topological,bal2023semiclassical}.

\medskip

The remainder of the paper is structured as follows. After introducing mathematical preliminaries on Dirac equations and layer potential operators, we construct a boundary integral formulation for \eqref{eq:Dirac} and analyze it in the case of a flat interface $\Gamma$ separating $\Omega_1$ from $\Omega_2$ in section \ref{sec:prelim}. Our main results on the boundary integral formulation for general interfaces $\Gamma$ and generic energies $E$ are stated in section \ref{sec:analytical}. These results are also generalized to the case of two interfaces separating three subdomains with constant mass terms. Several illustrations of the theory and accuracy of our approach are demonstrated in section \ref{sec:numerical}. Potential extensions of our results are briefly described in section \ref{sec:extensions} while the proofs of our main results are given in the Appendix.

\section{Mathematical preliminaries}\label{sec:prelim}
\subsection{Detailed formulation of the problem}\label{subsec:detailed}
In this section we give a precise formulation of the problem and introduce assumptions on the data and interface required by the analysis. Suppose we are given a smooth simple curve $\Gamma$ separating the plane into a lower region $\Omega_1$ and an upper region $\Omega_2.$ Let $\gamma:\mathbb{R}\to \mathbb{R}^2$ be an arclength parameterization. Moreover, with $\hat{n}(t)$ the normal vector to $\Gamma$ at $\gamma(t)\in\mathbb{R}^2$ pointing in the direction of $\Omega_2$, 
we assume that $(\gamma'(t),\hat n(t))$ has positive orientation. For concreteness, we additionally assume that $\gamma \in C^\infty(\mathbb{R};\mathbb{R}^2),$
\begin{align}\label{eq:beta}
|\gamma'(t)| =1, \qquad |\gamma^{(j)}(t)| \le C_j e^{-\dc |t|},\, j=2,3,4,\dots 
\end{align}
for some positive real constants $C_j$ and $\dc,$ and
\begin{align}
\lim_{t \to \infty} |\gamma(\pm t)| = \infty, \qquad \lim_{t\to \infty } |\gamma(t) - \gamma(-t)| =\infty. \label{eq:gammainfty}
\end{align}
\begin{remark}
    The assumption that $\gamma$ is smooth is purely for ease of exposition. The results in this paper can be extended easily to interfaces that are only $C^3$ with decay only in the second and third derivative of $\gamma.$ Only the regularity result, Lemma \ref{lem_boot}, would need to be adjusted accordingly.
\end{remark}
The fact that $\gamma$ is one-to-one (along with assumption \eqref{eq:gammainfty} above) 
implies the existence of some $c>0$ such that
\begin{align}\label{eq:c}
    \frac{|\gamma (t) -\gamma (s)|}{|t-s|} \ge c, \qquad s,t \in \mathbb{R}.
\end{align}
Let $m,E \in \mathbb{R}$ such that $0 < |E| < |m|$. Suppose $f_j \in \mathcal{C}^\infty_c (\Omega_j;\mathbb{C}^2)$ for $j=1,2$.
In this paper, we are concerned with solving
\begin{align}\label{eq:PDE}
\begin{split}
    (-i \partial_{x_1} \sigma_1- i \partial_{x_2} \sigma_2 + m \sigma_3 - E) u (x) = f_2 (x), \qquad &x \in \Omega_2\\
    (-i \partial_{x_1} \sigma_1- i \partial_{x_2} \sigma_2 - m \sigma_3 - E) u (x) = f_1 (x), \qquad &x \in \Omega_1\\
    \lim_{y\to x} u(y) = \lim_{z \to x} u(z), \qquad & y \in \Omega_2, z \in \Omega_1, x \in \Gamma
\end{split}
\end{align}
for $u$. 
We supplement \eqref{eq:PDE} with appropriate radiation conditions at infinity, namely that
\begin{align}\label{eq:out0}
\begin{split}
    \lim_{t \rightarrow \pm \infty} (\pm\partial_t - iE) u (\gamma (t) + r \hat{n} (t))&=(0,0), \qquad r \ne 0,\\
    \lim_{d(x,\Gamma) \rightarrow \infty} u (x) &= (0,0).
\end{split}
\end{align}
As we will see below, these radiation conditions are necessary to make our problem well posed. In words, these constraints state that $u$ must propagate \emph{outward} along $\Gamma$ (with frequency $E$) and decay away from $\Gamma$.

\subsection{Connection with the time-harmonic Klein-Gordon equation}\label{subsec:KG}
If we define $$H := -i \partial_{x_1} \sigma_3+ i \partial_{x_2} \sigma_2 + \tilde{m} \sigma_1,$$ where $\tilde{m} = \tilde{m} (x) = (-1)^j m$ whenever $x \in \Omega_j$, then, by a unitary transformation, \eqref{eq:PDE} reduces to $(H-E) u = f_1 + f_2$, where (with some abuse of notation) we have redefined $f_j \leftarrow U^* f_j$ and $u \leftarrow U^* u$, with $U := (\sigma_1 + \sigma_3)/\sqrt{2}$.
Applying $H+E$ to both sides, we obtain $(H^2 - E^2) u = \tilde{f}$, for some $\tilde{f} \in \mathcal{C}^\infty_c (\Omega_1 \cup \Omega_2; \mathbb{C}^2)$.
By the anti-commutation relations of the Pauli matrices, $\{\sigma_i, \sigma_j\} = 2 \delta_{ij},$ it follows that $(-\Delta+\omega^2) u = \tilde{f}$ in $\Omega_1 \cup \Omega_2$, where $\omega := \sqrt{m^2 - E^2}$. Moreover, derivatives of $u$ must satisfy
\begin{align*}
    [[(\hat{n} \cdot \nabla) u]](\gamma (t)) = 2m (\hat{n} \cdot \tilde{\sigma}) u (\gamma (t)),
\end{align*}
where $[[g]]$ denotes the jump of $g$ across $\Gamma$, $\hat n=\hat n(t)$, and $\tilde{\sigma}:= (\sigma_2, \sigma_3)$. 
When $\Gamma = \{x_2 = 0\}$ is a flat interface, we thus obtain the following decoupled system, 
\begin{align*}
    (-\Delta+\omega^2) u^{(j)} &= \tilde{f}^{(j)}, \qquad \qquad \qquad \qquad \Omega_1 \cup \Omega_2,\\
    [[u^{(j)}_{x_2}]](t,0) &= (-1)^{j-1} 2m u^{(j)} (t,0), \qquad t\in \mathbb{R},
\end{align*}
where $j=1,2$.

When $j=2$, we recover exactly the time-harmonic Klein-Gordon equation analyzed in \cite{bal2022integral}. The $j=1$ component corresponds to a solution that decays rapidly in all directions.

\subsection{A naive integral equation formulation}\label{subsec:naive}
A standard way to represent the solution of \eqref{eq:PDE} is by the decomposition $u = u_i + u_s$, where
\begin{align}\label{eq:ui}
    u_i (x) := \begin{cases}
        \int_{\Omega_2} G_2 (x,y) f_2 (y) {\rm d} y, & x \in \Omega_2\\
        \int_{\Omega_1} G_1 (x,y) f_1 (y) {\rm d} y, & x \in \Omega_1
    \end{cases}
\end{align}
is known as the ``incident field'' and 
\begin{align}\label{eq:us}
    u_s (x) := \begin{cases}
        \int_\mathbb{R} G_2 (x, \gamma (t)) \mu (t) {\rm d} t, & x \in \Omega_2\\
        \int_\mathbb{R} G_1 (x, \gamma (t)) \mu (t) {\rm d} t, & x \in \Omega_1
    \end{cases}
\end{align}
the ``scattered field,''
for some density $\mu : \mathbb{R} \rightarrow \mathbb{C}^2$.
Here, 
\begin{align}\label{eq:gf}
    G_j (x,y) := \frac{1}{2\pi} (-i \partial_{x_1} \sigma_1 -i \partial_{x_2} \sigma_2 + (-1)^j m \sigma_3 + E) K_0 (\omega |x-y|)
\end{align}
is the 
Green's function for the constant mass Dirac equation in $\Omega_j$, with $\omega := \sqrt{m^2 - E^2}$ and $K_0$ the modified Bessel function of the second kind. Note that since $\omega>0$, the Green's function decays exponentially as $|x-y|$ increases. 

This ansatz guarantees that any function $u$ decomposed in such a way solves the Dirac equation in $\Omega_1$ and $\Omega_2$. What remains is to choose the {\it density} $\mu$ such that $u$ is continuous across $\Gamma$, i.e., the jump in $u_s$ across $\Gamma$ must exactly cancel the jump in $u_i$. This jump in $u_s$ can be calculated analytically using standard properties of the above Green's functions. 
It is convenient to write $u_s = \usn + \usd$, where
\begin{align*}
    \usn (x) := \begin{cases}
        \frac{1}{2\pi} \int_{\mathbb{R}} (m \sigma_3 + E) K_0 (\omega |x-\gamma (t')|) \mu (t') {\rm d}t', & x \in \Omega_2\\
        \frac{1}{2\pi} \int_{\mathbb{R}} (-m \sigma_3 + E) K_0 (\omega |x-\gamma (t')|) \mu (t') {\rm d}t', & x \in \Omega_1
    \end{cases} 
\end{align*}
and
\begin{align*}
    \usd (x) := \frac{1}{2\pi} \int_{\mathbb{R}} (-i \partial_{x_1} \sigma_1 -i \partial_{x_2}\sigma_2) K_0 (\omega |x-\gamma (t')|) \mu (t') {\rm d} t', \qquad x \in \Omega_1 \cup \Omega_2.
\end{align*}

We define the \emph{single layer potential} $S_\omega$ by 
\begin{align}\label{eq:SL}
S_\omega [\mu] (x) := \frac{1}{2\pi}\int_{\mathbb{R}} K_0 (\omega |x-\gamma (t')|) \mu (t') {\rm d} t'.
\end{align}
It is well known that the function $x \mapsto S_\omega [\mu] (x)$ is continuous on all of $\mathbb{R}^2$, hence
the jump of 
$\usn$ across $\gamma (t)$ is
\begin{align*}
    \lim_{\substack{x \rightarrow \gamma (t) \\ x \in \Omega_2}} \usn (x) - \lim_{\substack{x \rightarrow \gamma (t) \\ x \in \Omega_1}} \usn (x)
    = 2m \sigma_3 \cS_\omega [\mu] (t),
\end{align*}
where the operator $\cS_\omega$ given by
\begin{align}\label{eq:cS}
    \cS_\omega [\mu] (t) := 
    \frac{1}{2\pi}\int_{\mathbb{R}} K_0 (\omega |\gamma (t)-\gamma (t')|) \mu (t') {\rm d} t' \,,
\end{align}
is the restriction of $S_{\omega}$ to the boundary. 
Indeed, the transition of the $\sigma_3$ term from $-m$ in $\Omega_1$ to $m$ in $\Omega_2$ is solely responsible for the discontinuity in $\usn$. 

For $\usd$, we use that the tangential and normal derivatives of the single layer potential \eqref{eq:SL} satisfy
\begin{align}\label{eq:tangradS}
    \lim_{s\to 0^+} \gamma'(t) \cdot (\nabla S_\omega[\rho](t+ s\,\hat{n}(t)) - \nabla S_\omega[\rho](t- s\,\hat{n}(t))) = 0
\end{align}
and
\begin{align}\label{eq:jumpgradS}
    \lim_{s\to 0^+} \hat{n}(t) \cdot \nabla S_\omega[\rho](t\pm s\,\hat{n}(t)) = \mp \frac{1}{2}\rho(t) +\frac{1}{2\pi}\int_\mathbb{R} \hat{n}(t) \cdot \nabla K_0 (\omega |\gamma(t)-\gamma(t')|) \,\rho(t')\, {\rm d}t',
\end{align}
where $\hat{n} (t)$ is the unit vector normal to $\Gamma$ at the point $\gamma (t)$;
see, e.g. \cite[Lemmas 3.3 and 3.5]{holzmann2019boundary}.
Therefore,
\begin{align*}
    \lim_{\substack{x \rightarrow \gamma (t) \\ x \in \Omega_2}} \usd (x) - \lim_{\substack{x \rightarrow \gamma (t) \\ x \in \Omega_1}} \usd (x)
    =i (\hat{n}_1 (t) \sigma_1 + \hat{n}_2 (t) \sigma_2) \mu (t).
\end{align*}
That is, the jump in $\usd$ results only from the component of $\frac{1}{2\pi}\nabla K_0$ normal to $\Gamma$, with the normal derivative producing a jump of $-\mu$.

Combining the jumps of $\usn$ and $\usd$, it follows that
\begin{align*}
    [[u_s]](\gamma (t))=
    2m \sigma_3 \cS_\omega [\mu](t) + i (\hat{n}(t) \cdot \sigma) \mu (t),
\end{align*}
where we use the shorthand $\sigma := (\sigma_1, \sigma_2)$ so that $\hat{n}(t) \cdot \sigma = \hat{n}_1 (t) \sigma_1 + \hat{n}_2 (t) \sigma_2$.
Therefore,
the density $\mu$ must satisfy
\begin{align*}
    2m \sigma_3 \cS_\omega [\mu](t) &+ i (\hat{n}(t) \cdot \sigma) \mu (t)=-[[u_i]] (\gamma (t)). 
\end{align*}
Multiplying both sides by $-i \hat{n} (t) \cdot \sigma$ and using the anti-commutation property of the Pauli matrices, $\sigma_i \sigma_j + \sigma_j \sigma_i = 2 \delta_{ij},$ we obtain
\begin{align}\label{eq:int0}
    (1 - 2m i (\hat{n}(t) \cdot \sigma) \sigma_3 \cS_\omega) [\mu](t)=i (\hat{n}(t) \cdot \sigma) [[u_i]] (\gamma (t)).
\end{align}
It will be convenient to solve an integral equation whose left-hand side is invariant to rotations of $\Gamma$. To this end, we define
\begin{align}\label{eq:cobmat}
\cobmat := \frac{1}{\sqrt{2}}\begin{pmatrix}
    1 & in_1 + n_2\\
    -in_1 + n_2 & -1
\end{pmatrix}.
\end{align}
We see that $\cobmat$ is a unitary matrix that in general depends on $t,$ the position along the interface. 
Moreover, it can be verified from standard properties of the Pauli matrices that 
\begin{align}\label{eq:vnv}
\cobmat^* (\hat{n} \cdot \sigma) \sigma_3 \cobmat = i\sigma_3
\end{align}
is independent of $\Gamma$.
We thus set
$\rhovar := \cobmat^* \mu$ and multiply both sides of \eqref{eq:int0} by $\cobmat^*$ to obtain
\begin{align}\label{eq:int}
    \cL [\rhovar] (t) := (1 - 2m \cobmat^* (i\hat{n} \cdot \sigma) \sigma_3 \cS_\omega \cobmat) [\rhovar] (t) = i \cobmat^* \hat{n} \cdot \sigma [[u_i]] (\gamma (t)).
\end{align}
The above defines an integral equation for $\rhovar$ with the desired rotation-invariance property and whose solution would yield a function $u = u_i + u_s$ that is continuous in $\mathbb{R}^2$.
However, as we will show 
in \eqref{eq:a1} below, the operator 
$\cL$ has absolutely continuous spectrum that passes through zero in the case of a flat interface. This means 
$\cL$ is in general not invertible on $L^2 (\mathbb{R}; \mathbb{C}^2)$, 
so that a naive discretization of the integral equation as stated here
would not yield a viable means to solve the Dirac equation \eqref{eq:PDE}.

The reason for the presence of such continuous spectrum is that modes are allowed to propagate along the interface $\Gamma$. As for any such spectral problems, outgoing boundary conditions need to be imposed on \eqref{eq:PDE} in order to uniquely characterize the propagating solution.

\subsection{Analysis of the flat case}\label{subsec:flat}
In this section, we analyze the integral equation \eqref{eq:int} in the case of a flat interface. This allows us to {\em define} an appropriate notion of outgoing conditions.
By standard Fourier techniques, we show that the integral operator on the left-hand side has an absolutely continuous spectrum containing zero. We then propose a remedy involving an integral operator 
that implements the appropriate outgoing radiation conditions \eqref{eq:out0}
and yields a well-posed boundary integral formulation of the Dirac equation \eqref{eq:PDE}.

Suppose that the interface parameterized by $\gamma (t) = (t,0)$.
We first observe that $\cobmat$ is independent of $t$ 
and therefore commutes with $\cS_\omega$. It then follows from \eqref{eq:vnv} that $$\cL = 1 - 2m \cobmat^* (i\hat{n} \cdot \sigma) \sigma_3 \cobmat \cS_\omega = 1+2m \sigma_3 \cS_\omega.$$
We then have the simplification
\begin{align*}
    \cS_\omega [\mu] (t) =
    \frac{1}{2\pi}\int_{\mathbb{R}} K_0 (\omega |t-t'|) \mu (t') {\rm d} t',
\end{align*}
which leads to
\begin{align}\label{eq:FcS}
    \mathcal{F}_{t \rightarrow \xi} \{\cS_\omega [\mu] (t)\} =
    \frac{1}{2 \sqrt{\xi^2+\omega^2}} \tilde{\mu} (\xi),
\end{align}
where $\tilde{\mu}$ denotes the Fourier transform of $\mu$,  see~\cite{abramowitz}, for example.
It follows that
\begin{align}\label{eq:a1}
    \mathcal{F}_{t \rightarrow \xi} \{\cL[\rhovar] (t)\} = \Big(1 + \frac{m}{\sqrt{\xi^2+m^2-E^2}} \sigma_3\Big) \tilde{\rhovar} (\xi) =: a_1(\xi) \tilde{\rhovar} (\xi).
\end{align}
Since 
$a_1$ has continuous strictly monotonic (for $\xi>0$ and $\xi<0$) components and one of them vanishes when $\xi = \pm E$,
the operator $\cL$ has (purely absolutely) continuous spectrum containing zero and therefore cannot be inverted in $L^2 (\mathbb{R}; \mathbb{C}^2)$.

To resolve this issue, we define the operator
\begin{align}\label{eq:R}
    \cR \rho (t) := \frac{1}{2iE} \int_{\Rm} e^{iE|t-t'|} \rho (t'){\rm d}t', \qquad \rho \in L^2 (\mathbb{R}),
\end{align}
as the {\em outgoing} inverse of the one-dimensional Helmholtz operator
\begin{align}\label{eq:1dgf}
    (\partial^2_t + E^2) \cR [\rho] (t) = \rho (t), \qquad \rho \in L^2 (\mathbb{R}).
\end{align}
We then set
$\mmat:=\mmat_1$ if $m<0$ and $\mmat:=\mmat_2$ if $m>0$, where
\begin{align}\label{eqn:mdef}
\mmat_1 = \begin{pmatrix}
        1 & 0\\
        0 & 0
    \end{pmatrix}
\quad {\rm and}\quad 
\mmat_2 = \begin{pmatrix}
        0 & 0\\
        0 & 1
\end{pmatrix}.
\end{align}
Define 
$\cP := 1 + 2im^2 \mmat \cR$, so that
\begin{align}\label{eq:cP}
    \cP \rho (t) &= \rho (t) + \mmat \frac{m^2}E \int_{\Rm} e^{iE|t-t'|} \rho (t'){\rm d}t'
\end{align}
for $\rho \in L^2 (\mathbb{R}; \mathbb{C}^2)$.

As was the case for $\cL$, the operator 
$\cR$ acts as a convolution and so is 
easily analyzed in the Fourier domain.
A direct calculation yields
\begin{align*}
    \cF_{t \rightarrow \xi} \Big\{\frac{1}{2iE}e^{iE|t|}\Big\} = \frac{1}{2iE}\lim_{\eps \downarrow 0}\cF_{t \rightarrow \xi} \Big\{e^{i(E+i\eps)|t|}\Big\}=-\frac{1}{\xi^2-E^2}
\end{align*}
(which could of course have also been obtained directly from \eqref{eq:1dgf}), and
hence
\begin{align*}
    \mathcal{F}_{t \rightarrow \xi} \{\cR[\rho] (t)\} = -\frac{1}{\xi^2-E^2}\tilde{\rho} (\xi).
\end{align*}
We conclude that
$
    \cF_{t \rightarrow \xi} \{\cL \cP [\rho] (t)\} = a_1 (\xi) a_2 (\xi) \tilde{\rho} (\xi),
$
where $a_1$ is defined by \eqref{eq:a1} and
\begin{align}\label{eq:a2}
    a_2 (\xi) := 1 - \mmat \frac{2i m^2}{\xi^2-E^2}.
\end{align}
It can easily be verified that $a(\xi) := a_1 (\xi) a_2 (\xi)$ defines an analytic (matrix-valued) function whose eigenvalues are bounded and bounded away from zero, with
\begin{align*}
    a^{-1} (\xi) =
        (1-\mmat) \frac{\sqrt{\xi^2+\omega^2}}{\sqrt{\xi^2+\omega^2} + |m|} +\mmat \Big(1+\frac{|m|}{\sqrt{\xi^2+\omega^2}} \Big) \Big( 1 + \frac{(2i+1)m^2}{\xi^2-E^2-2im^2} \Big).
\end{align*}
Therefore, $\cL \cP$ is invertible on $L^2 (\mathbb{R}; \mathbb{C}^2)$, so that the integral equation
\begin{align}\label{eq:solverhoflat}
    \cL\cP [\rho] (t) =i \cobmat^* \hat{n} \cdot \sigma [[u_i]] (\gamma (t))
\end{align}
has a unique solution $\rho$.
We can then set $\rhovar := \cP \rho$ to obtain a solution of \eqref{eq:int}.
A solution of \eqref{eq:PDE} is then obtained by setting $\mu := \cobmat \rhovar$ and $u := u_i + u_s$, where $u_i$ and $u_s$ are defined by \eqref{eq:ui} and \eqref{eq:us}, with the latter dependent on $\mu$.

With Theorem \ref{thm:usol} below, we will establish that the above solution $u$ also satisfies the radiation conditions \eqref{eq:out0}. 
If $\cR$ were to be replaced by 
\begin{align*}
    \Rvar \rho (t) :=\frac{-1}{2iE} \int_{\Rm} e^{-iE|t-t'|} \rho (t'){\rm d}t', \qquad \rho \in L^2 (\mathbb{R}),
\end{align*}
then the same procedure would yield a different (incoming) solution of \eqref{eq:PDE}; namely, one that satisfies the adjoint boundary condition 
\begin{align*}
    \lim_{t \rightarrow \pm \infty} (\pm\partial_t + iE) u (\gamma (t) + r \hat{n} (t))&=0, \qquad r \ne 0.
\end{align*}

\section{Analytical results}\label{sec:analytical}

\subsection{The boundary integral formulation}\label{subsec:bie}
In Section \ref{subsec:flat}, we derived a well-posed boundary integral formulation \eqref{eq:solverhoflat} of the Dirac equation \eqref{eq:PDE} in the case of a flat interface. 
Appealing to this derivation, we now define the boundary integral formulation of \eqref{eq:PDE} for the general (non-flat interface) case by \eqref{eq:solverhoflat} as well.
That is, with $\cL$ and $\cP$ respectively given by \eqref{eq:int} and \eqref{eq:cP}, our boundary integral formulation of \eqref{eq:PDE} is to find a density $\rho \in L^2 (\mathbb{R}; \mathbb{C}^2)$ such that 
\begin{align}\label{eq:solverho}
    \cL\cP [\rho] (t) =i \cobmat^* \hat{n} \cdot \sigma [[u_i]] (\gamma (t)).
\end{align}
Recall the definitions \eqref{eq:ui} and \eqref{eq:cobmat} of $u_i$ and $\cobmat$, and that $\hat{n} \cdot \sigma := n_1 \sigma_1 + n_2 \sigma_2$, where the $\sigma_j$ are the Pauli matrices \eqref{eq:Pauli} and $(n_1, n_2) = (n_1 (t), n_2(t)) = \hat{n} (t)$ is the unit vector normal to $\Gamma$ at $\gamma (t)$ pointing towards $\Omega_2$.

Once $\rho$ is obtained from \eqref{eq:solverho}, a solution $u$ of \eqref{eq:PDE} is then constructed from $\rho$, as 
will be made rigorous with Theorem \ref{thm:usol} below.
We first establish the general solvability of the integral equation \eqref{eq:solverho} with Theorems \ref{thm:invE} and \ref{thm:invL}.
These theorems 
also guarantee an exponential decay of the solution $\rho$, which will be quantified by the following weighted $L^2$ space.

For $\alpha\in \mathbb{R}$, define $w_\alpha (t) := e^{\alpha|t|}$ and $L^2_\alpha := \{\rho \in L^2 (\mathbb{R}) : w_\alpha \rho \in L^2(\mathbb{R})\}$.
\begin{theorem}\label{thm:invE}
Fix $m\ne 0$.
For any $\eps>0$ and
$\alpha>0$ sufficiently small (depending on $\eps$),
the integral equation \eqref{eq:solverho} admits a unique solution $\rho \in L^2_\alpha (\mathbb{R}; \mathbb{C}^2)$
for all but a finite number of $E \in [-|m|+\eps, -\eps] \cup [\eps, |m|-\eps]$.
\end{theorem}
\begin{theorem}\label{thm:invL}
Fix $m_0, E_0 \in \mathbb{R}$ such that $0 < |E_0| < |m_0|$, and set 
$m=\lambda m_0$ and $E = \lambda E_0$ for $\lambda \in \mathbb{R}$.
Then for any $\alpha > 0$ sufficiently small,
the integral equation \eqref{eq:solverho} admits a unique solution $\rho \in L^2_\alpha(\mathbb{R}; \mathbb{C}^2)$
for all but a finite number of $\lambda \in [1,\infty)$.
\end{theorem}
The constructed solution of \eqref{eq:solverho} in $L^2_\alpha$ is in fact smoother when the interface $\gamma(t)$ is smooth. Under assumption \eqref{eq:beta}, we have the following regularity result:
\begin{theorem}\label{thm:regul}
    Let $f_j\in \mathcal{C}^\infty_c (\Omega_j)$ for $j=1,2$. Then the solution $\rho$ constructed in Theorem \ref{thm:invL} satisfies $\rho \in (L^2_\alpha\cap C^\infty)(\mathbb{R};\mathbb{C}^2)$.
\end{theorem}
A slightly more precise version of this theorem is proved in Lemma \ref{lem_boot}.


We now show how the solution $\rho$ of our integral equation \eqref{eq:solverho} can be used to construct a solution of the Dirac equation satisfying natural radiation conditions at infinity along the interface $\Gamma$. Although it would be natural to conclude that the Dirac equations with such radiation conditions admits a unique solution, this problem is not considered here.
\begin{theorem}\label{thm:usol}
    Let $f_j\in \mathcal{C}^\infty_c (\Omega_j)$ for $j=1,2$. Let $u_{i}$ be the incident field as defined in~\eqref{eq:ui}, and $u_{s}$ the scattered field as defined in~\eqref{eq:us}. 
    Suppose $\rho \in L^2_\alpha \cap C^\infty (\mathbb{R}; \mathbb{C}^2)$ satisfies \eqref{eq:solverho}, and let $\mu := \cobmat \cP \rho$.
 Then $u=u_{i}+u_{s}$ satisfies~\eqref{eq:PDE} and \eqref{eq:out0}.
\end{theorem}



\subsection{Two interfaces}\label{sec:two}
Suppose that instead we are given two non-intersecting smooth simple curves, $\Gamma_1$ and $\Gamma_2$, separating the plane into a lower region $\Omega_0$, a middle region $\Omega_1$, and an upper region $\Omega_2$. We refer to Figure \ref{fig:scattering} (top right panel) for an illustration. 
Let $\gamma_1, \gamma_2:\mathbb{R}\to \mathbb{R}^2$ be the arclength parameterizations of the interfaces $\Gamma_1$ and $\Gamma_2$. Moreover, with $\hat{n}_j(t)$ the unit normal vector to $\Gamma_j$ at $\gamma_j (t)$ 
pointing in the direction of $\Omega_{j}$, 
we assume that $(\gamma_j'(t),\hat n_j(t))$ has positive orientation. We assume that for $j=1,2$, $\gamma = \gamma_j$ satisfies \eqref{eq:beta} and \eqref{eq:gammainfty}.
In addition, we assume that
\begin{align}\label{eq:gammainfty2}
    \lim_{t \rightarrow \infty} |\gamma_2 (\eps_2 t) - \gamma_1 (\eps_1 t)| = \infty
\end{align}
for all $\eps_1, \eps_2 \in \{-1,1\}$.
That is, the two interfaces must be going off to infinity in different directions.

With $m,E \in \mathbb{R}$ satisfying $0 < |E| < |m|$ and $f_j \in \mathcal{C}_c^\infty (\Omega_j)$ for $j=0,1,2$, we wish to solve
\begin{align}\label{eq:PDE2}
\begin{split}
    (-i \partial_{x_1} \sigma_1- i \partial_{x_2} \sigma_2 + m \sigma_3 - E) u (x) = f_2 (x), \qquad &x \in \Omega_2\\
    (-i \partial_{x_1} \sigma_1- i \partial_{x_2} \sigma_2 - m \sigma_3 - E) u (x) = f_1 (x), \qquad &x \in \Omega_1\\
    (-i \partial_{x_1} \sigma_1- i \partial_{x_2} \sigma_2 + m \sigma_3 - E) u (x) = f_0 (x), \qquad &x \in \Omega_0\\
    \lim_{y\to x} u(y) = \lim_{z \to x} u(z), \qquad & y \in \Omega_2, z \in \Omega_1, x \in \Gamma_2\\
    \lim_{y\to x} u(y) = \lim_{z \to x} u(z), \qquad & y \in \Omega_1, z \in \Omega_0, x \in \Gamma_1
\end{split}
\end{align}
for $u$, 
with the following radiation conditions at infinity:
\begin{align}\label{eq:out2}
\begin{split}
    \lim_{t \rightarrow \pm \infty} (\pm\partial_t - iE) u (\gamma_j (t) + r \hat{n}_j (t))&=0, \qquad j=1,2, \quad r \ne 0,\\
    \lim_{d(x,\Gamma) \rightarrow \infty} u (x) &= 0,
\end{split}
\end{align}
where we have defined $\Gamma := \Gamma_1 \cup \Gamma_2$.
Define $G_1$ and $G_2$ by \eqref{eq:gf}, and set
\begin{align}
\label{eq:ui2}
    u_i (x) = \begin{cases}
        \int_{\Omega_2} G_2 (x,y) f_2 (y) {\rm d} y, & x \in \Omega_2\\
        \int_{\Omega_1} G_1 (x,y) f_1 (y) {\rm d} y, & x \in \Omega_1\\
        \int_{\Omega_0} G_2 (x,y) f_0 (y) {\rm d} y, & x \in \Omega_0
    \end{cases}
\end{align}
and
\begin{align}
\label{eq:us2}
    u_s (x) = \begin{cases}
        \int_\mathbb{R} G_2 (x, \gamma_2 (t)) \mu_2 (t) {\rm d} t, & x \in \Omega_2\\
        \int_\mathbb{R} G_1 (x, \gamma_2 (t)) \mu_2 (t) {\rm d} t
        + \int_\mathbb{R} G_1 (x, \gamma_1 (t)) \mu_1 (t) {\rm d} t, & x \in \Omega_1\\
        \int_\mathbb{R} G_2 (x, \gamma_1 (t)) \mu_1 (t) {\rm d} t, & x \in \Omega_0
    \end{cases}
\end{align}
for some $\mu_1, \mu_2 \in L^2 (\mathbb{R};\mathbb{C}^2)$.
Then $u:=u_i+u_s$ solves \eqref{eq:PDE2}; all that is left is to ensure that $u$ is continuous across $\Gamma_1$ and $\Gamma_2$.
Let $[[u]]_j$ denote the jump in $u$ across $\Gamma_j$. Using the arguments from Section \ref{subsec:naive}, we see that
\begin{align*}
    [[u_s]]_2(\gamma_2 (t))&= (2m \sigma_3 \cS_\omega + i \hat{n}_2 \cdot \sigma) \mu_2 (t) - \int_{\mathbb{R}} G_1 (\gamma_2 (t),\gamma_1 (t')) \mu_1 (t') {\rm d} t',\\
    [[u_s]]_1(\gamma_1 (t))&= (-2m \sigma_3 \cS_\omega + i \hat{n}_1 \cdot \sigma) \mu_1 (t) + \int_{\mathbb{R}} G_1 (\gamma_1 (t),\gamma_2 (t')) \mu_2 (t') {\rm d} t'.
\end{align*}
Writing $[[u_s]]_j = -[[u_i]]_j$ and multiplying both sides by $-i \hat{n}_j \cdot \sigma$, we get
\begin{align}\label{eq:system}
    \begin{pmatrix}
        1+2m (i \hat{n}_1 \cdot \sigma)\sigma_3 \cS_\omega & (-i\hat{n}_1 \cdot \sigma)\bar{\cK}_2\\
        (i \hat{n}_2 \cdot \sigma) \bar{\cK}_1 & 1-2m (i \hat{n}_2 \cdot \sigma)\sigma_3 \cS_\omega
    \end{pmatrix}
    \begin{pmatrix}
        \mu_1\\
        \mu_2
    \end{pmatrix}
    =
    \begin{pmatrix}
        (i \hat{n}_1 \cdot \sigma) [[u_i]]_1\\
        (i \hat{n}_2 \cdot \sigma) [[u_i]]_2
    \end{pmatrix},
\end{align}
where
\begin{align}\label{eq:bar_cK}
    \bar{\cK}_j [\mu] (t) := \int_{\mathbb{R}} G_1 (\gamma_i (t),\gamma_j (t')) \mu (t') {\rm d} t', \qquad i:=3-j, \qquad j \in \{1,2\}.
\end{align}
Observe that equation \eqref{eq:system} is the two-interface analogue of \eqref{eq:int0}. As before, we seek to write our integral equation such that the operator on the left-hand side is rotationally invariant.
For $j=1,2$, define $$\cobmat_j := \frac{1}{\sqrt{2}}\begin{pmatrix}
    1 & in^{(1)}_j + n^{(2)}_j\\
    -in^{(1)}_j + n^{(2)}_j & -1
\end{pmatrix},$$
with $n^{(i)}_j$ the $i$th component of the vector $\hat{n}_j$.
Set $\cobmat_j \rhovar_j := \mu_j$ and multiply both sides of \eqref{eq:system} by 
$\begin{pmatrix}
    \cobmat_1^* & 0\\
    0 & \cobmat_2^*
\end{pmatrix}$ 
to obtain
\begin{align}\label{eq:systemrotinv}
    \begin{pmatrix}
        \cL_1 & \cK_2\\
        \cK_1 & \cL_2
    \end{pmatrix}
    \begin{pmatrix}
        \tau_1\\
        \tau_2
    \end{pmatrix} =
    \begin{pmatrix}
        \cobmat_1^*(i \hat{n}_1 \cdot \sigma) [[u_i]]_1\\
        \cobmat_2^*(i \hat{n}_2 \cdot \sigma) [[u_i]]_2
    \end{pmatrix},
\end{align}
where
\begin{align*}
    \cL_j &:= 1- (-1)^j 2m \cobmat_j^* (i \hat{n}_j \cdot \sigma)\sigma_3 \cS_\omega \cobmat_j
\\
    \cK_j &:= (-1)^{j-1} \cobmat_\ell^* (i \hat{n}_\ell \cdot \sigma) \bar{\cK}_j \cobmat_j, \qquad \ell := 3 - j, \qquad j \in \{1,2\}.
\end{align*}
The operator
$\begin{pmatrix}
    \cL_1 & \cK_2\\
    \cK_1 & \cL_2
\end{pmatrix}$
on the left-hand side of
\eqref{eq:systemrotinv} is now invariant with respect to rotations of $\Gamma$. 
But as was the case 
for one interface (recall \eqref{eq:int} and the paragraph below it), this operator cannot be inverted in $L^2(\mathbb{R}; \mathbb{C}^4)$, and therefore the integral equation must be modified and outgoing conditions selected.
Motivated by our derivation in Section \ref{subsec:flat},
we will thus solve
\begin{align}\label{eq:solverho2}
    \begin{pmatrix}
        \cL_1 & \cK_2\\
        \cK_1 & \cL_2
    \end{pmatrix}
    \begin{pmatrix}
        \cP_1 & 0\\
        0 & \cP_2
    \end{pmatrix}
    \begin{pmatrix}
        \rho_1\\
        \rho_2
    \end{pmatrix} =
    \begin{pmatrix}
        \cobmat_1^*(i \hat{n}_1 \cdot \sigma) [[u_i]]_1\\
        \cobmat_2^*(i \hat{n}_2 \cdot \sigma) [[u_i]]_2
    \end{pmatrix}
\end{align}
for $(\rho_1, \rho_2)$, where
$\cP_1 := 1 + 2im^2 (1-\mmat) \cR$ and $\cP_2 = 1 + 2im^2 \mmat\cR$, 
where $\cR$ is defined in~\eqref{eq:R}, and the matrix $\mmat$ is given in~\eqref{eqn:mdef}

We are now ready to extend Theorems \ref{thm:invE} and \ref{thm:invL} to this two-interface setting, thus establishing that
the integral equation \eqref{eq:solverho2} in general has a unique solution.
\begin{theorem}\label{thm:invE2}
Fix $m\ne 0$.
For any $\eps>0$ and
$\alpha>0$ sufficiently small (depending on $\eps$),
the integral equation \eqref{eq:solverho2} admits a unique solution $(\rho_1, \rho_2) \in L^2_\alpha (\mathbb{R}; \mathbb{C}^4)$
for all but a finite number of $E \in [-|m|+\eps, -\eps] \cup [\eps, |m|-\eps]$.
\end{theorem}
\begin{theorem}\label{thm:invL2}
Fix $m_0, E_0 \in \mathbb{R}$ such that $0 < |E_0| < |m_0|$, and set 
$m=\lambda m_0$ and $E = \lambda E_0$ for $\lambda \in \mathbb{R}$.
Then for any $\alpha > 0$ sufficiently small,
the integral equation \eqref{eq:solverho2} admits a unique solution $(\rho_1,\rho_2) \in L^2_\alpha (\mathbb{R}; \mathbb{C}^4)$
for all but a finite number of $\lambda \in [1,\infty)$.
\end{theorem} 
When the interface is smooth and all sources are compactly supported, then a generalization of Theorem \ref{thm:regul} shows that each $\rho_j\in (L^2_\alpha \cap C^\infty) (\mathbb{R};\mathbb{C}^4)$. The solution $(\rho_1, \rho_2)$ of \eqref{eq:solverho2} can then be used to generate a solution of the two-interface Dirac equation \eqref{eq:PDE2} satisfying the appropriate radiation conditions \eqref{eq:out2}.
\begin{theorem}\label{thm:usol2}
    Suppose $(\rho_1, \rho_2) \in (L^2_\alpha\cap C^\infty) (\mathbb{R}; \mathbb{C}^4)$ satisfies \eqref{eq:solverho2}, set $\mu_j := \cobmat_j \cP_j \rho_j$ for $j=1,2$, and define $u := u_i + u_s$, where $u_i$ and $u_s$ are respectively given by \eqref{eq:ui2} and \eqref{eq:us2}, with the latter depending on $(\mu_1, \mu_2)$. Then \eqref{eq:PDE2} and \eqref{eq:out2} hold.
\end{theorem}

\section{Numerical examples}\label{sec:numerical}
In this paper, the integral operators $\cL,$ $\cK,$ and $\cP$ were discretized using the boundary integral equation package \texttt{chunkie} \cite{chunkie}, which uses piecewise Legendre polynomial expansions to represent boundary curves and densities. The action of integral operators on densities is computed using a mixture of standard Gauss-Legendre quadrature~ (\cite{bremer2010universal,bremer2010nonlinear}) for well-separated points and smooth kernels and specialized {\it generalized Gaussian quadrature} rules (for nearby points and weakly-singular kernels). We note that as written, the domains of definition of the operators $\cL,$ $\cK,$ and $\cP$ appearing in our equations are functions defined on all of $\mathbb{R}.$ The decay of the data and solutions to the integral equation justifies truncating $\cP$ to functions supported in a region $I_1$ containing the support of the data and the non-flat sections of the geometry and a suitable `buffer region' which scales logarithmically in the required accuracy. Its image is functions supported on a larger interval $I_2$. Analogously, the domain of $\cL$ and $\cK$ can be reduced to maps from functions supported on $I_2$ to functions supported on $I_1.$ The application of these operators is accelerated using a combination of fast multipole methods~\cite{rokhlin1990rapid,greengard1987fast,fmm2d} and sweeping algorithms~\cite{jiang2022}.

\begin{figure}[ht!]
    \centering
    \includegraphics[scale=.12]{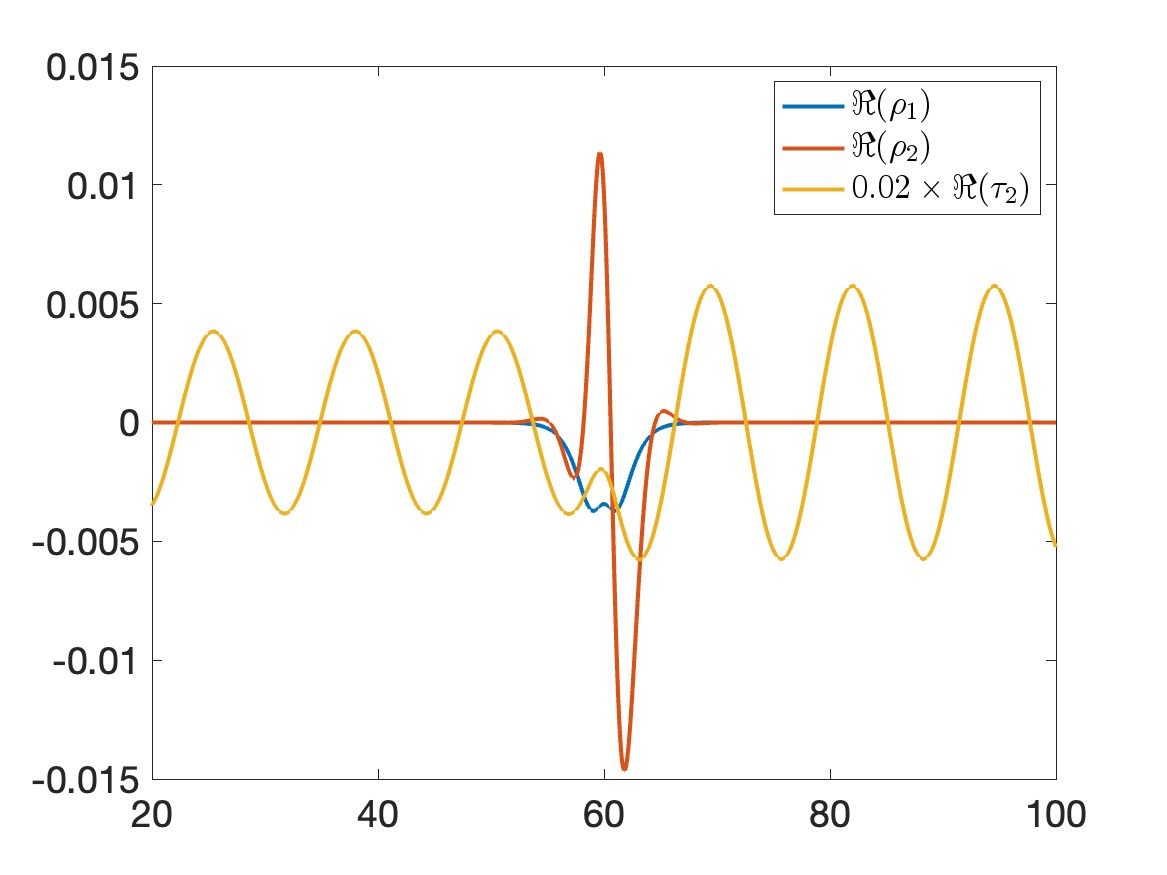}
    \includegraphics[scale=.12]{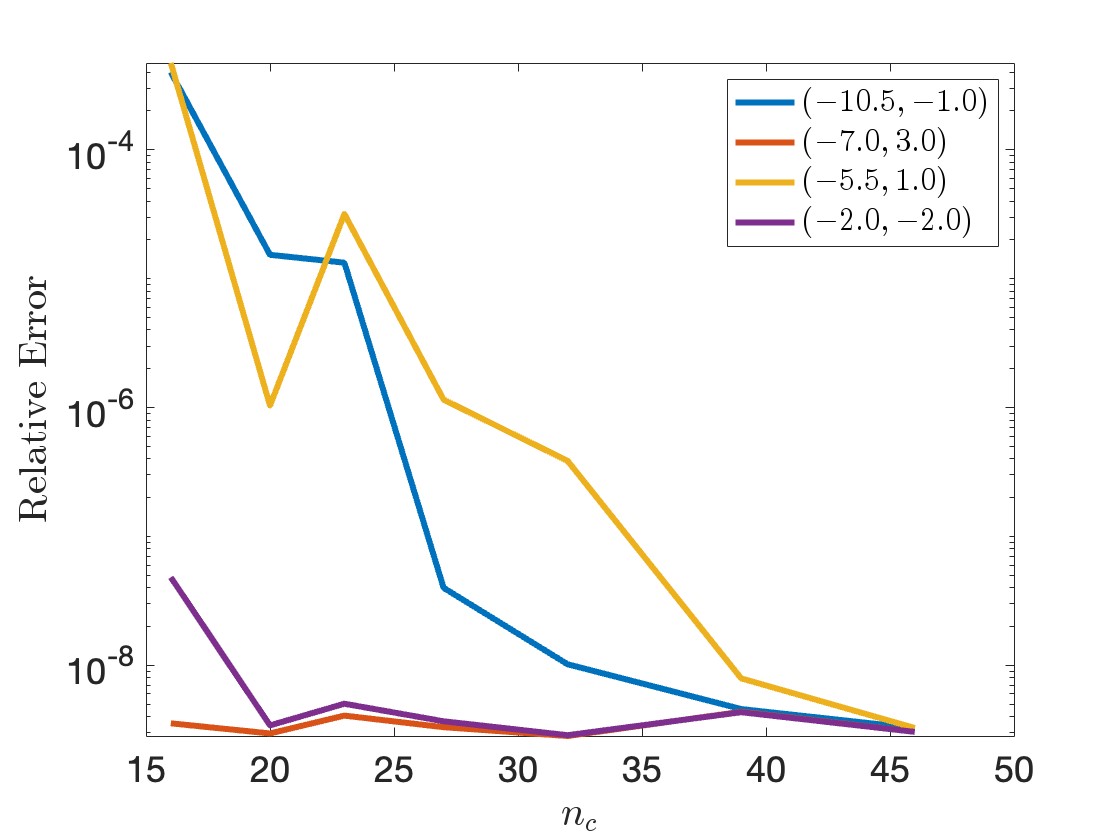}\\
    \includegraphics[scale=.12]{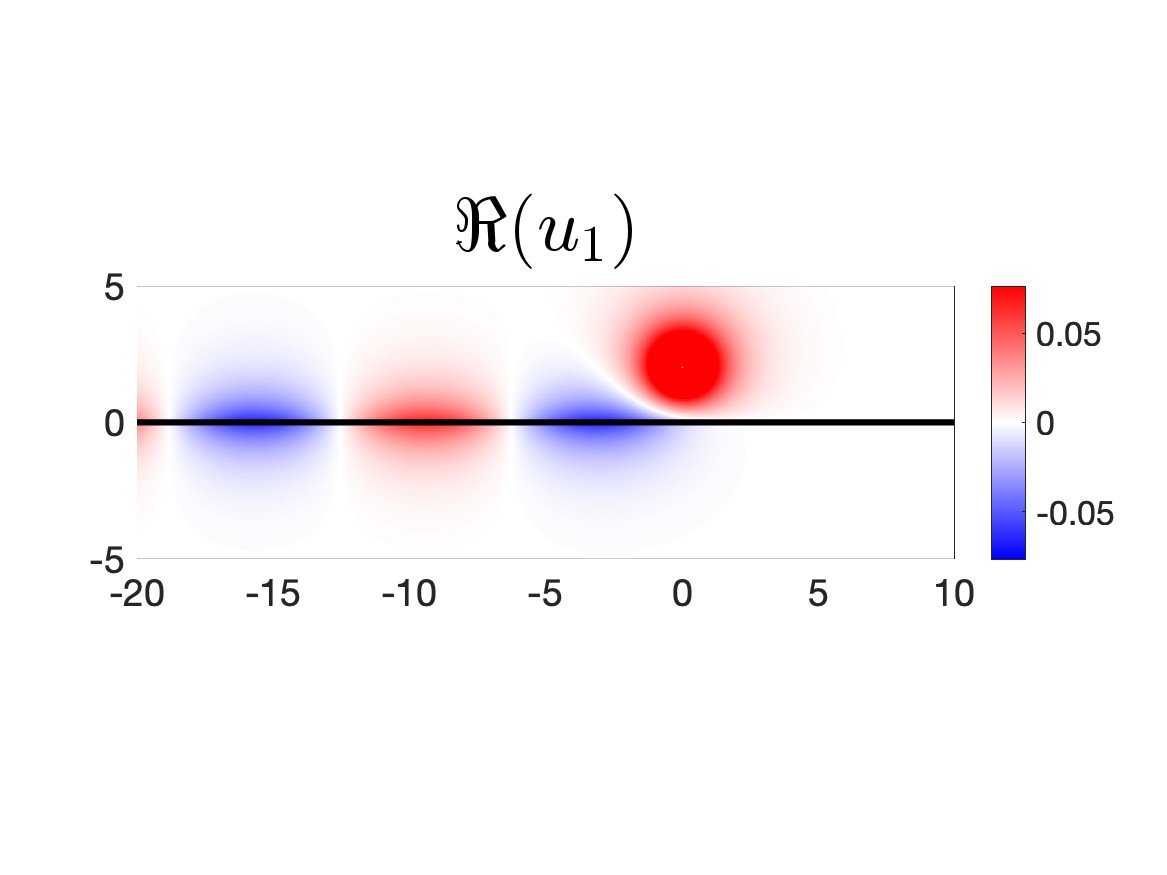}
    \includegraphics[scale=.12]{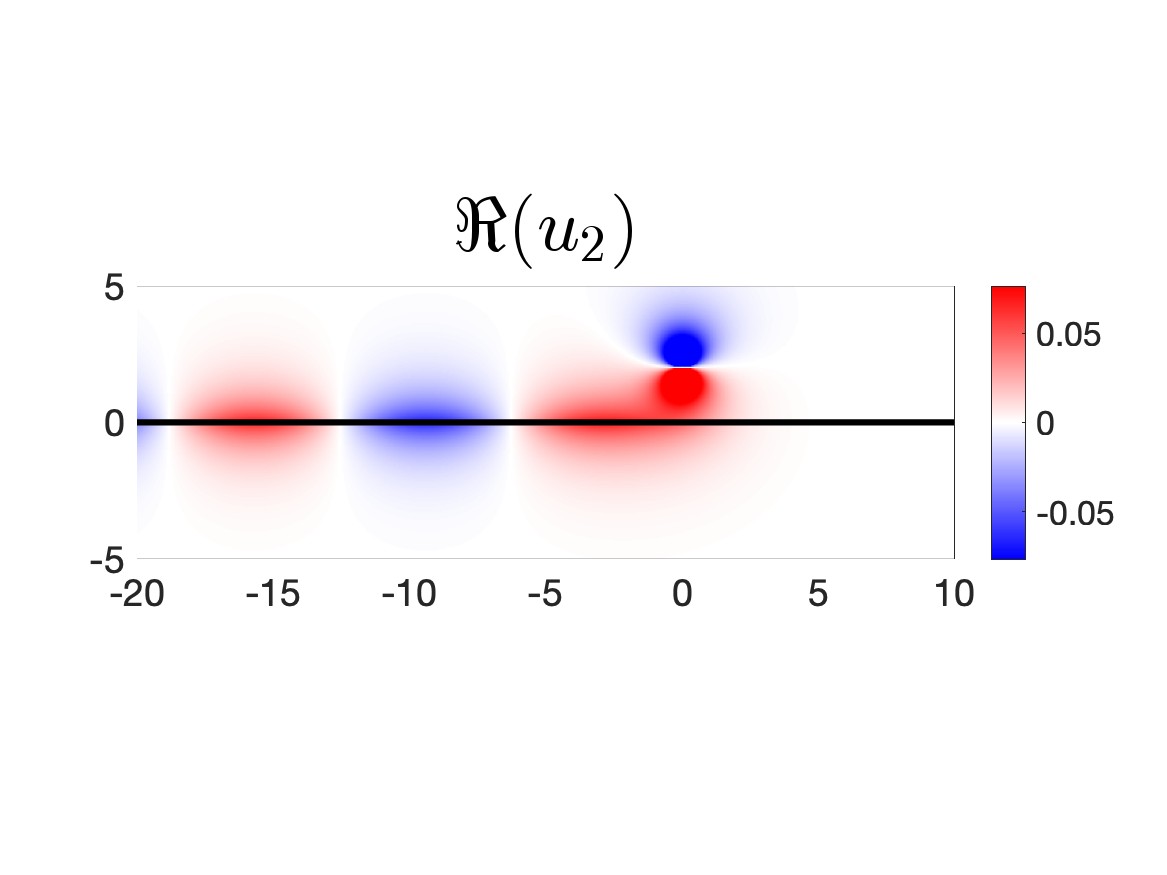}\\
    \includegraphics[scale=.12]{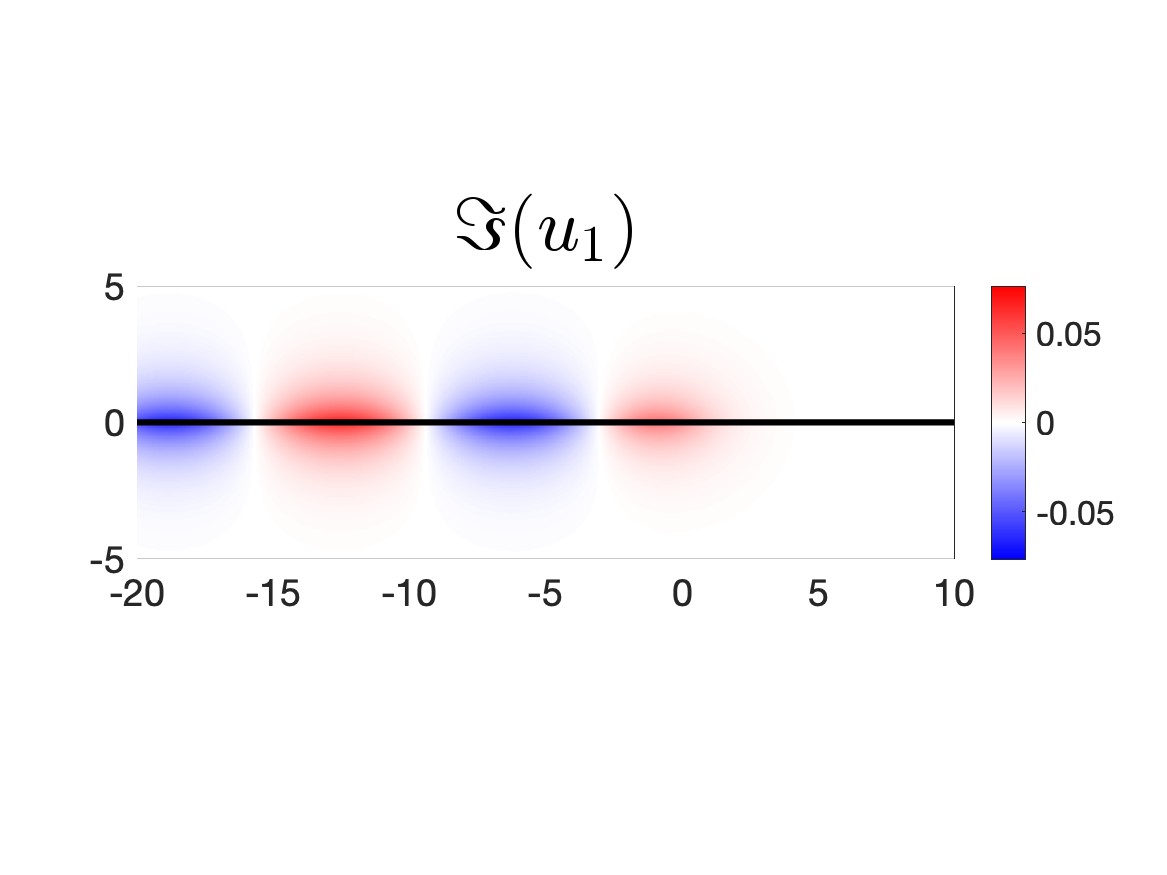}
    \includegraphics[scale=.12]{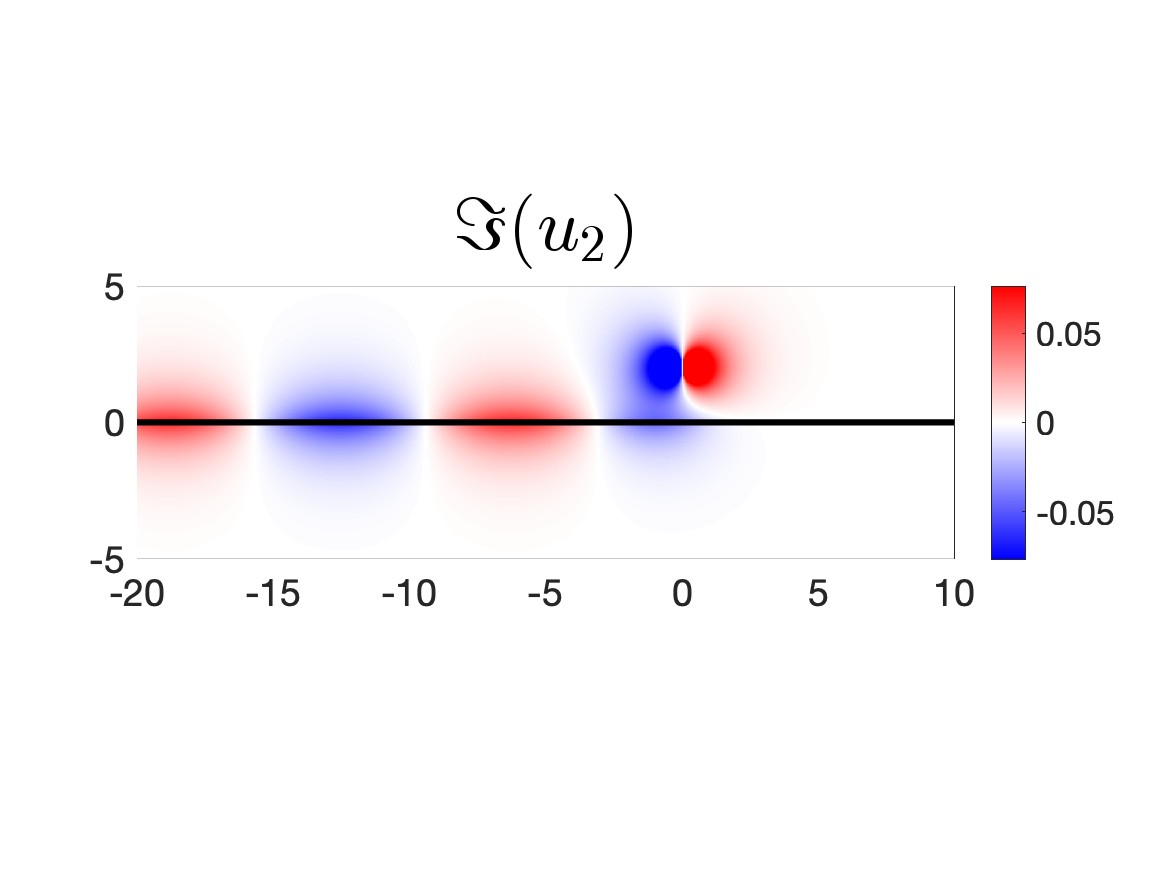}
    \caption{Solution for the flat interface, $\Gamma = \{x_2 = 0\}$, where $m=1, E=0.5$ and the right-hand side is $f_1 = (0,0), f_2 = (\delta_{x_0}, 0)$ with source location $x_0 = (0,2)$. The full solution is illustrated by the second and third rows. The top left panel plots the real parts of the computed densities $\rho$ and $\tau$. Recall that $\tau_1 = \rho_1$. The top right panel contains the relative error of the solution $u_1$ at the four specified points, as a function of the number of discretization points.}
    \label{fig:flat}
\end{figure}

We now illustrate the performance of our method through several numerical experiments. As predicted by the theory, solutions of the Dirac equation \eqref{eq:PDE} propagate only in one direction along $\Gamma$, with the waves traveling in the direction such that the region with positive mass is on the right. We refer to Figure \ref{fig:flat} for a flat-interface example, where we illustrate the computed densities $\rho$ and $\tau := \cP \rho$ (recall the definitions \eqref{eq:solverhoflat} and \eqref{eq:cP} of $\rho$ and $\cP$) as well as the full solution $u$. We also include a convergence check (top right panel), which demonstrates the high accuracy of our numerical method. Note that the relative error is evaluated at four distinct points and computed using the analytic solution for the flat interface. 

Observe that, although the solution $u$ decays rapidly to the right of the source (representing waves that propagate only from right to left in the time-dependent picture), the density $\tau$ propagates in both directions. Indeed, the interface for Figure \ref{fig:flat} is parametrized by $\gamma (t) = (t-60, 0)$, so that $t=60$ in the top-left panel is the point along $\Gamma$ that is closest to the source. In this case, the amplitude of $\tau$ is actually \emph{larger} to the right of the source, thus the construction of the ``scattered field'' $u_s$ prevents any rightward propagation. 
The cancellation of the right-propagating mode can be verified analytically when the interface is flat, though we do not carry out this derivation here. A natural question is whether one can construct an integral representation of the solution for which the propagating density propagates only to the left of the source. We postpone a further study of this issue to a future paper.

\begin{figure}
    \centering
    \includegraphics[scale=.10]{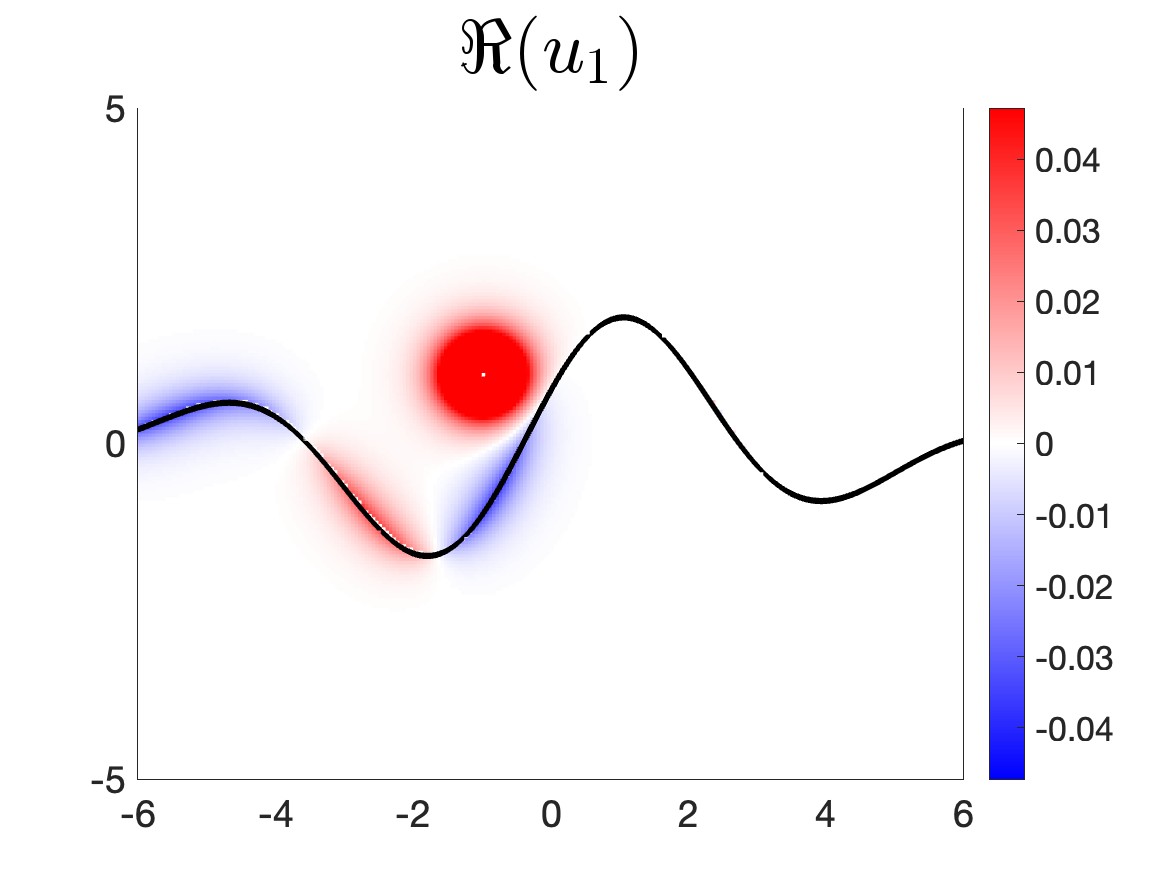}
    \includegraphics[scale=.10]{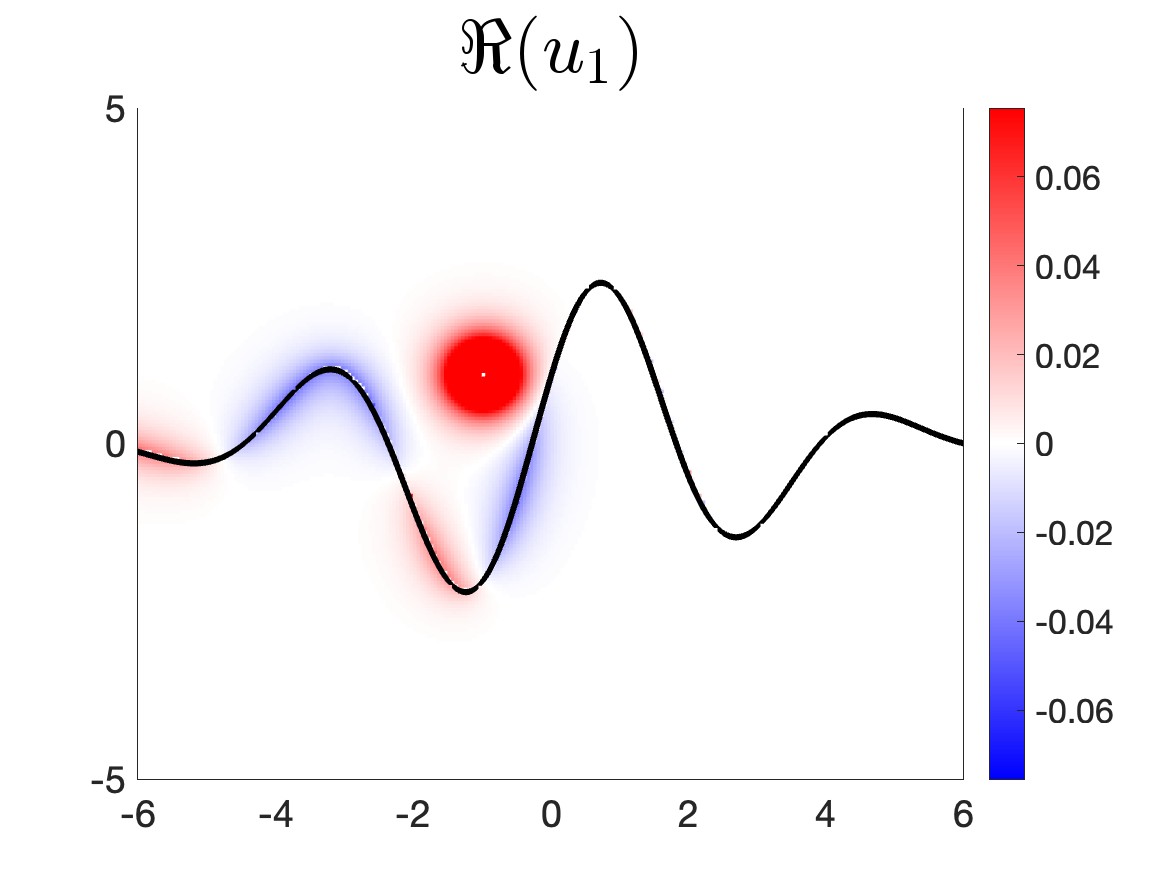}
    \includegraphics[scale=.10]{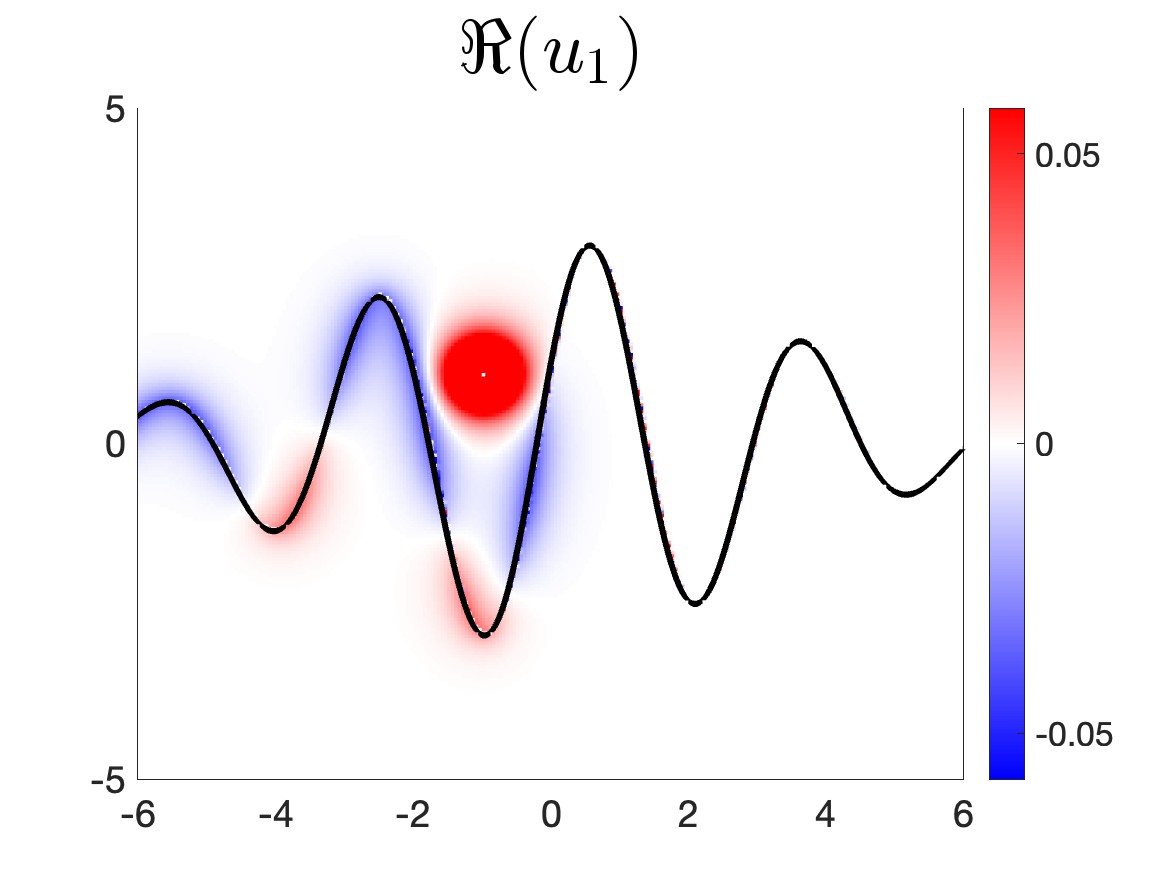}
    \caption{A component of the solution $u$ of the Dirac equation \eqref{eq:PDE} for various interfaces $\Gamma$. In each plot, the interface is illustrated by the solid curve. These examples all correspond to $m=4, E=1$ and a right-hand side given by $f_1 = (0,0), f_2 = (\delta_{x_0}, 0)$, with source location of $x_0 = (-1,1)$.}
    \label{fig:all}
\end{figure}

The asymmetric transport along $\Gamma$ persists even when the latter is highly oscillatory. We refer to Figure \ref{fig:all} for some examples. 
Note that this stability of asymmetric transport for the Dirac equation does not 
extend to the related Klein-Gordon equation,
\begin{align}\label{eq:KG_numerical}
\begin{split}
    (-\Delta+\omega^2) u (x) = f_1 (x), \qquad &x \in \Omega_1,\\
    (-\Delta+\omega^2) u (x) = f_2 (x), \qquad &x \in \Omega_2,\\
    [[\hat{n} \cdot \nabla u]] (\gamma (t)) = -2|m|u (\gamma(t)), \qquad &t \in \mathbb{R},
\end{split}
\end{align}
discussed in Section \ref{subsec:KG}; see \cite{bal2022integral} for more details.
We refer to Figure \ref{fig:resonance} for an example. 
While the Dirac solutions for $E=0.82$ and $E=0.83$ understandably look the same, the Klein-Gordon solution exhibits a visible change of behavior from a wave that mostly gets backscattered by the circular cavity ($E=0.82$) to one that 
partially transmits ($E=0.83)$.
More quantitatively, the Klein Gordon solutions for $E=0.82$ and $E=0.83$ have respective transmission coefficients of $9.8529 \times 10^{-4}$ and $0.3863$, while the transmission coefficient for the Dirac solution is of course $1$ for any energy.

\begin{figure}[t]
    \centering
    \includegraphics[scale=.16]{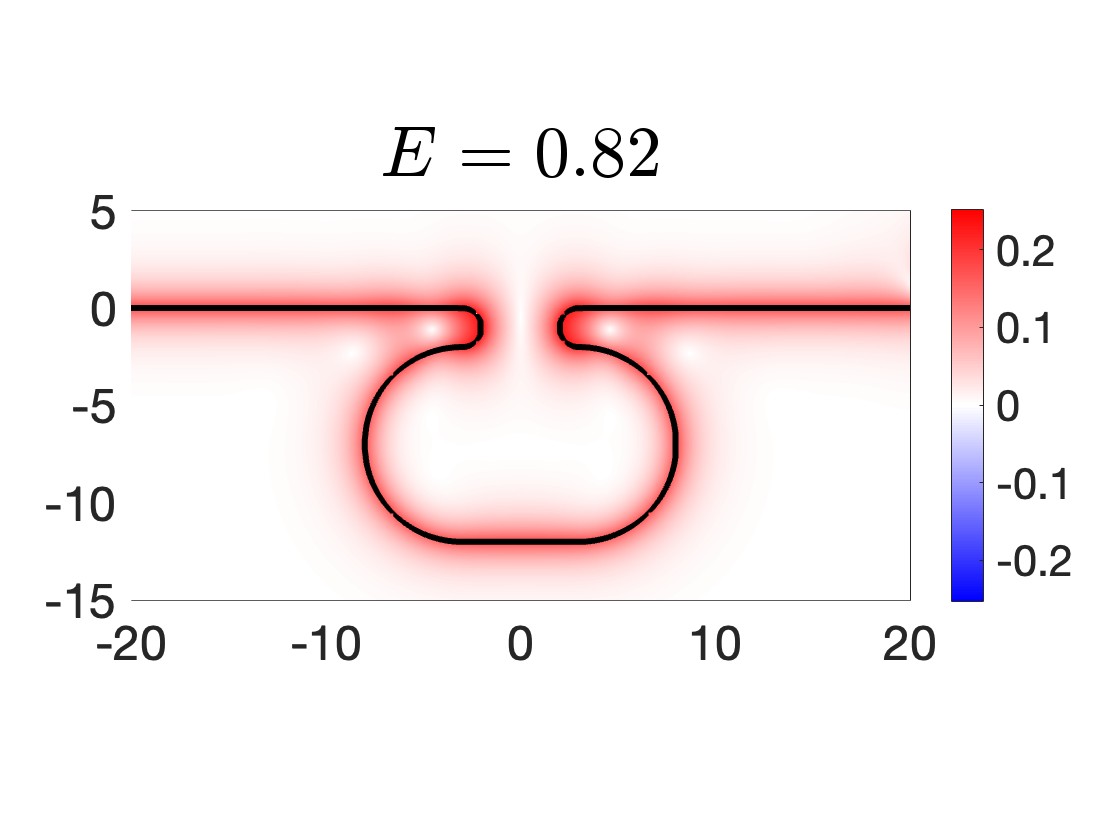}
    \includegraphics[scale=.16]{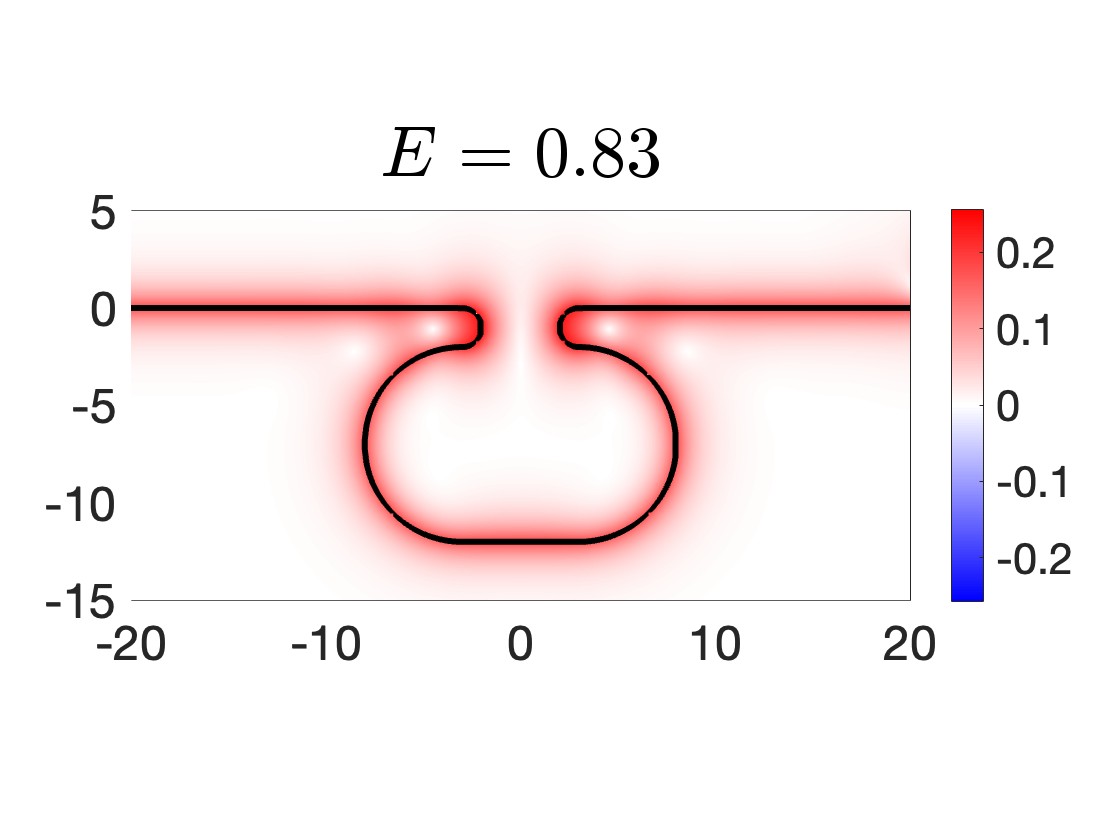}\\[-1.0cm]
    \includegraphics[scale=.16]{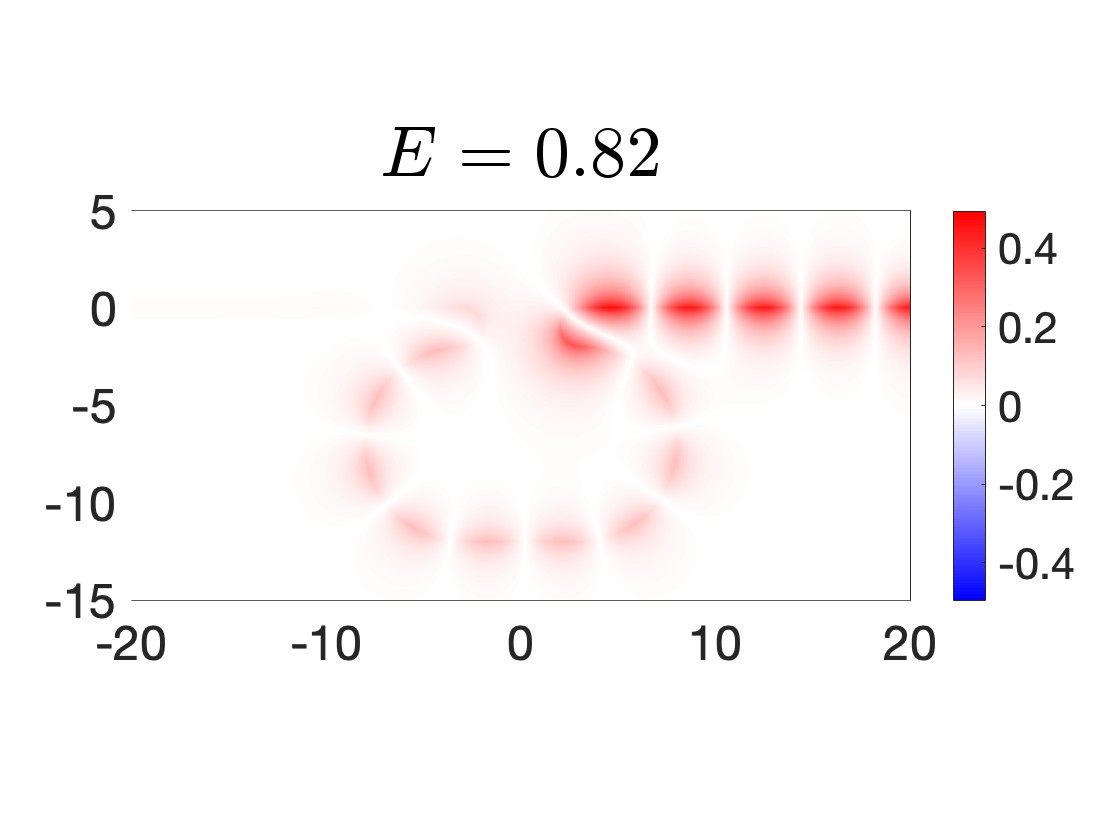}
    \includegraphics[scale=.16]{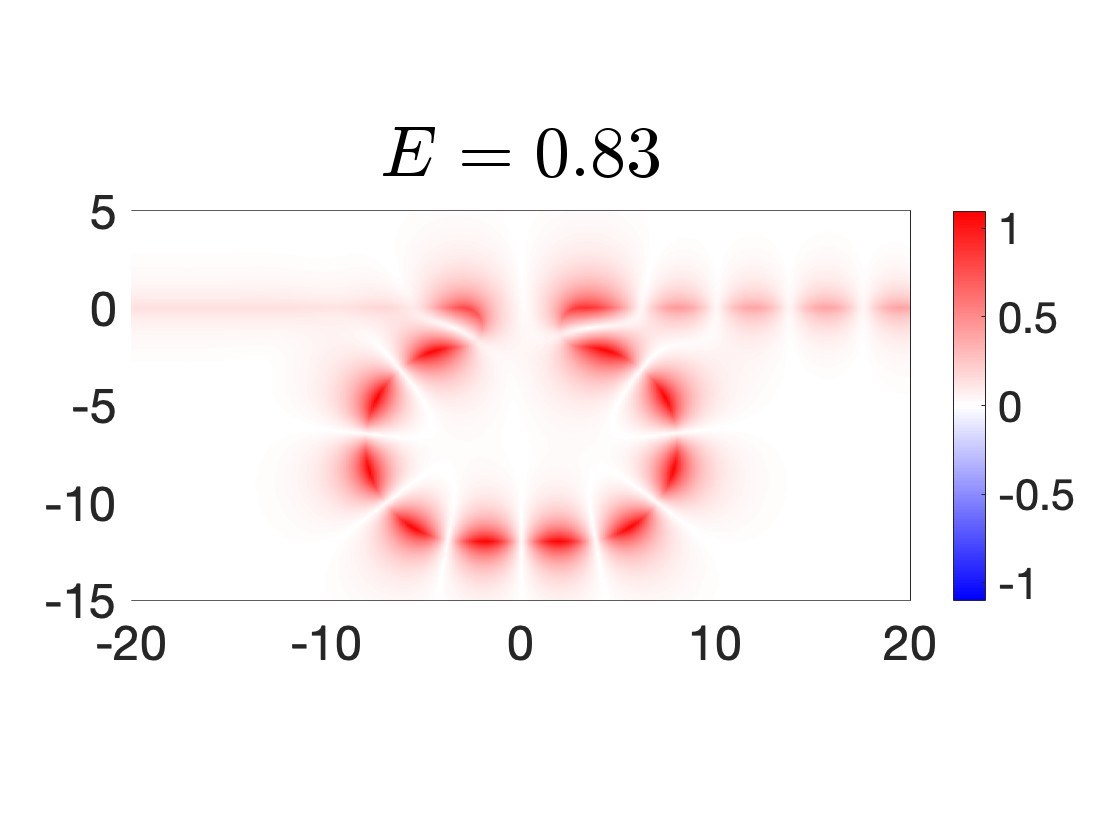}
    \caption{
    Top row: absolute value of (the first component of) the solution of the Dirac equation \eqref{eq:PDE}, with $m=1$ and two similar values of $E$. The interface is illustrated by the solid curve. The right-hand side is given by $f_1 = (0,0)$ and $f_2 = (\delta_{x_0},0)$, where the source location $x_0 = (22,1)$ is just above the interface. The solutions in the left and right panels are nearly identical. Bottom row: Absolute value of the solution of the corresponding Klein-Gordon equation \eqref{eq:KG_numerical}, with the same values of $m$ and $E$, and right-hand side $f_1 = 0, f_2 = \delta_{x_0}$. As opposed to the top row, the two solutions on the bottom row look qualitatively different; the signal at $E=0.82$ gets reflected while at $E=0.83$, it partially transmits.}
    \label{fig:resonance}
\end{figure}

It is natural to expect that this fundamental difference between the Dirac and Klein-Gordon equations could manifest itself in the spectral properties of the corresponding integral operators.
A solution that gets (partially) trapped in some bounded domain would likely correspond to an energy $E_*$ that is close to a resonance in the complex plane.


As a result, we would expect that for any fixed $x \in \mathbb{R}^2$, the function $E \mapsto u(x;E)$ would oscillate rapidly near such exceptional energy values $E_*$, where $u(x;E)$ denotes the solution of the Dirac equation \eqref{eq:PDE} with energy $E$.
But since solutions of the Dirac equation cannot get trapped (as discussed above), 
such exceptional values $E_*$ might not exist at all. On the other hand, solutions of the Klein-Gordon equation \eqref{eq:KG_numerical} 
can back-scatter, meaning that more of these exceptional energy values are likely to exist for the Klein-Gordon operator. 

\begin{figure}
    \centering
    \includegraphics[scale=.16]{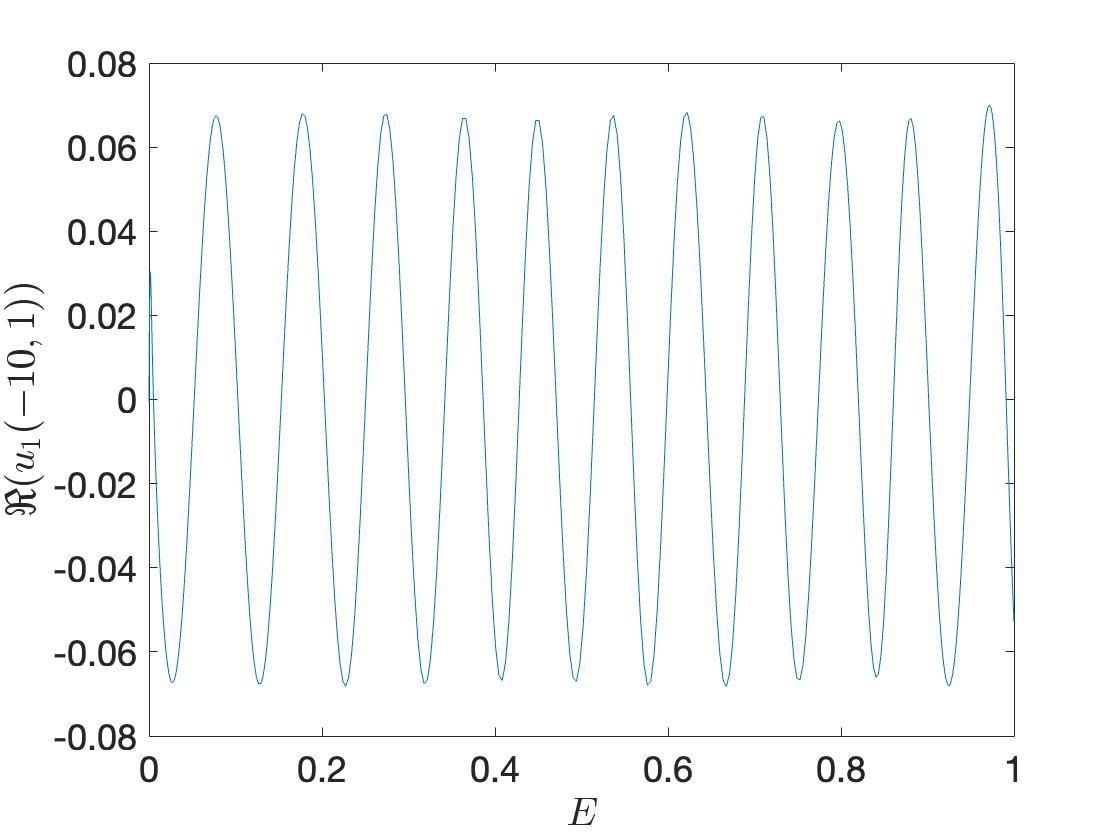}
    \includegraphics[scale=.16]{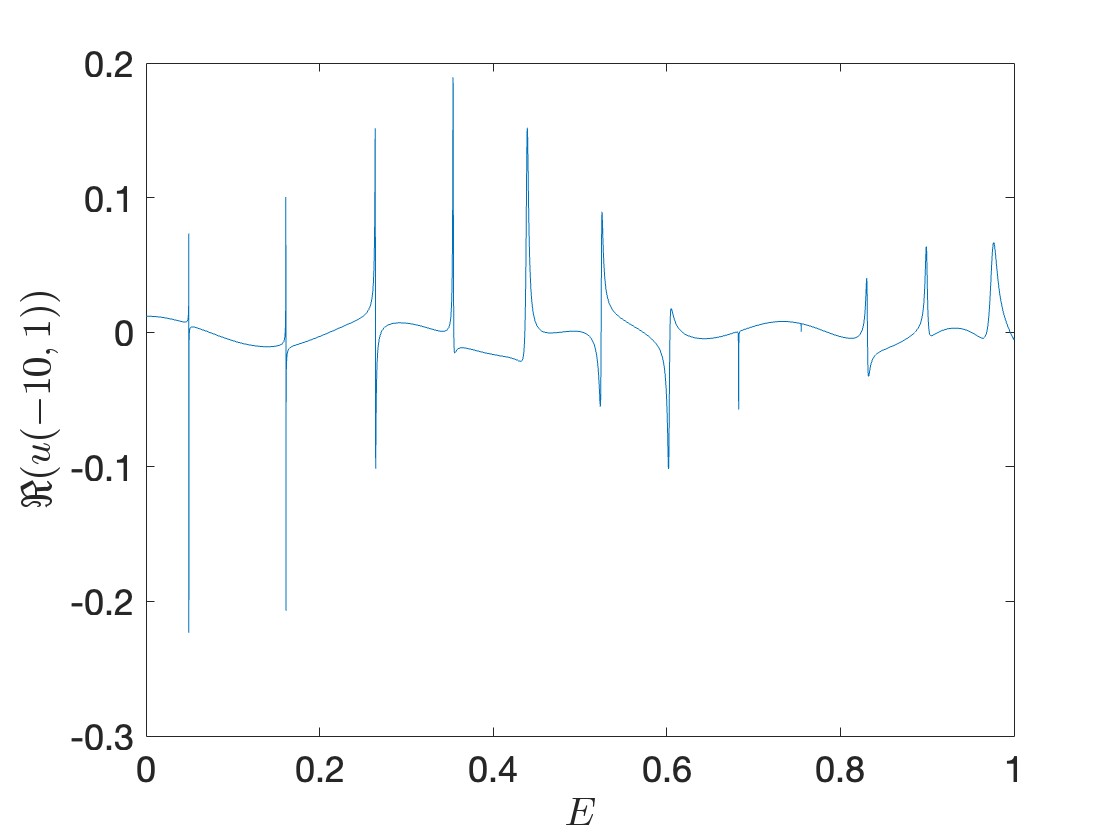}
    \caption{Left: a component of the solution of the Dirac equation as a function of $0<E<1$ for $m=1$, interface as in Figure \ref{fig:resonance}, and right-hand side $f_1 = (0,0), f_2 = (\delta_{x_0},0), x_0 = (22,1)$.  Right: the corresponding Klein-Gordon solution with right-hand side $f_1=0, f_2=\delta_{x_0}$.}
    \label{fig:sweep_E}.
\end{figure}

We support these predictions with a numerical example. The left panel of Figure \ref{fig:sweep_E} shows a component of the Dirac solution $u$ 
as a function of $E$, with $m=1$ fixed and interface given by Figure \ref{fig:resonance}. The solution $u$ is evaluated at the point $(-10,1)$, which lies just above the interface.
As expected, the solution is smooth in $E$ with no resonance detected.
The right panel shows the corresponding Klein-Gordon solution, which exhibits several sharp peaks.
This direct comparison between the Dirac and Klein-Gordon solutions illustrates that there are resonances near the real axis for the Klein-Gordon equation that are absent from the corresponding Dirac equation.
\medskip

We provide numerical examples of the two-interface problem \eqref{eq:PDE2}; see Figures \ref{fig:scattering} and \ref{fig:scattering_vdelta}.
The two interfaces are illustrated by the solid curves in the plots, and we let $d$ denote the distance between the two curves.
In each example, the source is placed a fixed distance away from the left branch $\Gamma_1$.
One can verify analytically that for any $\delta \ne 0$,
\begin{align*}
    \lim_{t \rightarrow -\infty} u(\gamma_1(t) + \delta \hat{n}_1 (t)) &=0\\
    \lim_{t \rightarrow \infty}
    \Big\{ u(\gamma_1(t) + \delta \hat{n}_1 (t)) - m e^{-m|\delta|} e^{iEt} \tilde{\rho}_1^{(1)} (E) \frac{1}{\sqrt{2}} \begin{pmatrix}
        1\\
        e^{i\phi_1}
    \end{pmatrix} \Big\} &= 0\\
    \lim_{t \rightarrow -\infty}
    \Big\{ u(\gamma_2(t) + \delta \hat{n}_2 (t)) - m e^{-m|\delta|} e^{-iEt} \tilde{\rho}_2^{(2)} (-E) \frac{1}{\sqrt{2}} \begin{pmatrix}
        e^{-i\phi_2}\\
        -1
    \end{pmatrix} \Big\} &= 0\\
    \lim_{t \rightarrow \infty} u(\gamma_2(t) + \delta \hat{n}_2 (t)) &=0
\end{align*}
where $\phi_1$ (resp. $\phi_2$) is the angle between the right branch of $\Gamma_1$ (resp. left branch of $\Gamma_2$) and the horizontal axis, $\tilde{\rho}_j^{(j)}$ denotes the Fourier transform of the $j$th entry of $\rho_j$,
and we have assumed that $m>0$ for concreteness.
It follows that the coefficient
\begin{align*}
    T_R := \frac{|\tilde{\rho}_1^{(1)} (E)|^2}{|\tilde{\rho}_1^{(1)} (E)|^2+|\tilde{\rho}_2^{(2)} (-E)|^2}
\end{align*}
measures the amount of signal that propagates from left to right along $\Gamma_1$. As expected, $T_R$ is monotonically increasing in $d \in (0, 1]$, as demonstrated by the bottom-right panels of Fig. \ref{fig:scattering} and \ref{fig:scattering_vdelta}.

In Fig. \ref{fig:scattering}, we analyze the effect of interface separation by holding the top interface ($\Gamma_2$) fixed while translating the bottom interface ($\Gamma_1$) in the vertical ($x_2$) direction. For small $d$, the interfaces are nearly tangent to one another at the origin and $T_R$ becomes very small (but not vanishing) in the limit $d\to0$ ($\approx 10^{-5}$). 
%
%
In Figure \ref{fig:scattering_vdelta}, the interfaces $\Gamma_j$ are parametrized by 
\begin{align}\label{eq:vdelta}
    \gamma_j (t) = \left( \frac{t}{\sqrt{1+a^2}},\, (-1)^j a \sqrt{\frac{t^2}{1+a^2}+d^2}\, \right),
\end{align}
with $a = 1/2$ fixed. In this case, the interfaces are still separated by the distance $d$, though they are no longer tangent in the $d \to 0$ limit (with $\Gamma_j$ resembling the graph of $x \mapsto (-1)^j a |x|$). As a result, $T_R$ converges to a larger value as $d\to0$.  

One could also ask how $\lim_{d \to 0} T_R$ depends on the slope value $a$. In Figure \ref{fig:t_vs_theta}, we plot this limit as a function of the angle $\theta$ of the top right branch of $\Gamma$ with respect to the horizontal (that is, $\tan \theta = a$). As expected, the limit is equal to $0.5$ when $\theta=\pi/4$ (indicated by the dashed lines in the figure). This limit models beam splitters, which have been analyzed in a number of contexts \cite{hammer2013dynamics,qiao2014current,bal2023mathematical}.


\begin{figure}[ht!]
    \centering
    \includegraphics[scale=0.16]{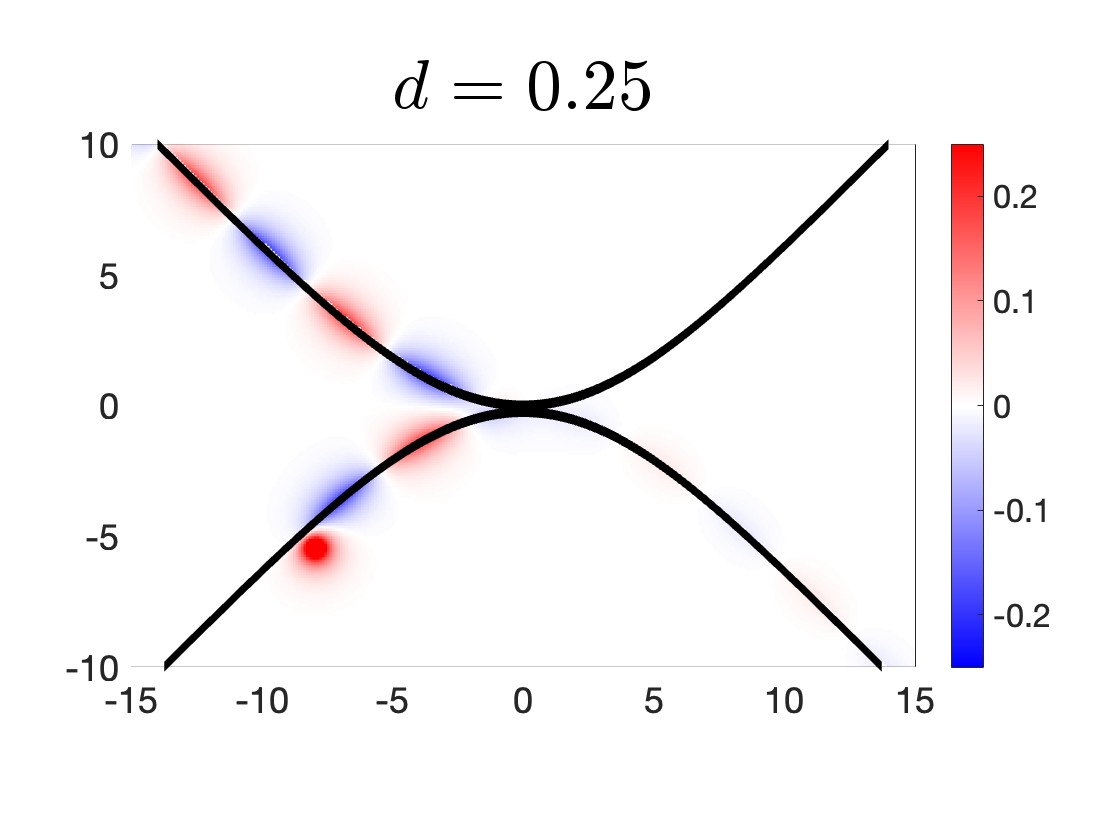}
    \includegraphics[scale=0.16]{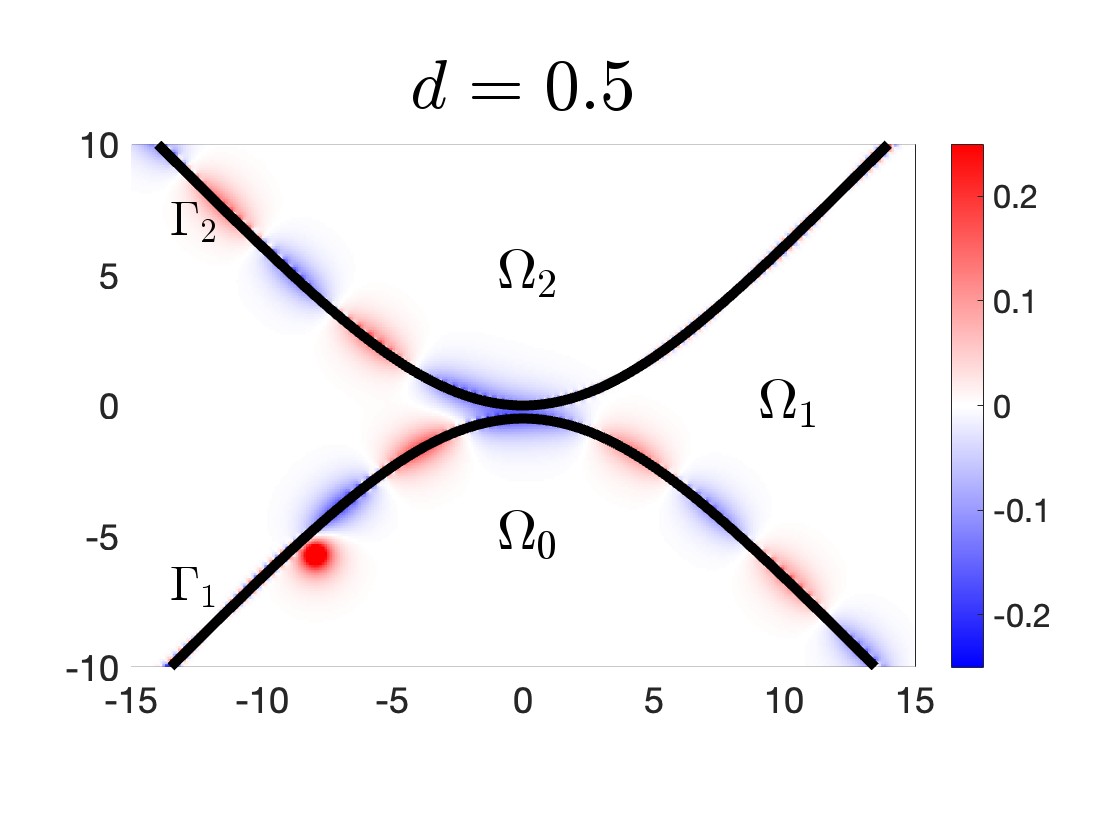}
    \includegraphics[scale=0.16]{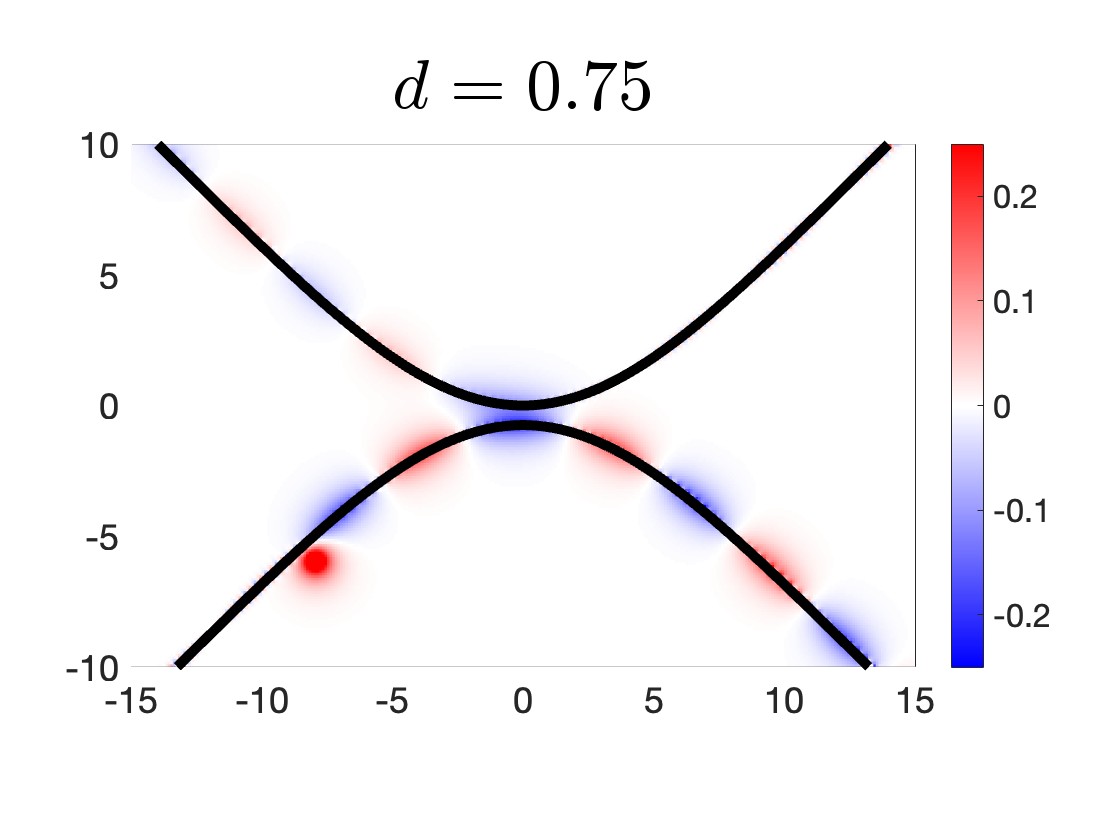}
    \includegraphics[scale=0.16]{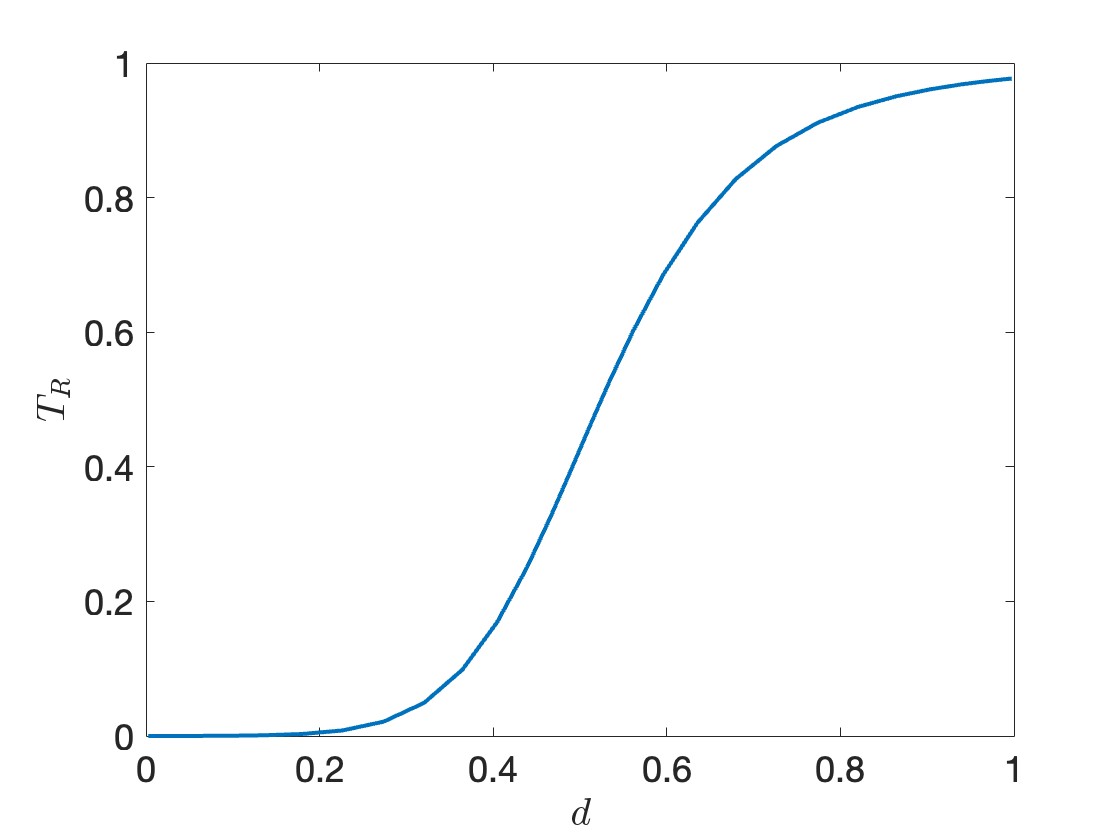}
    \caption{Top row and bottom left: plots of $\Re (u_1)$ for the two-interface problem with various interface separations $d$. Bottom right panel: The transmission coefficient $T_R$ as a function of $d$. In these examples, $m=2$, $E=0.8$, and the right-hand side is given by $f_0 = (\delta_{x_0},0), f_1 = (0,0), f_2 = (0,0)$ with source location 
    $x_0$ at the center of the red dot (at a fixed distance away from $\Gamma_1$).}
    \label{fig:scattering}
\end{figure}

\begin{figure}[ht!]
    \centering
    \includegraphics[scale=0.16]{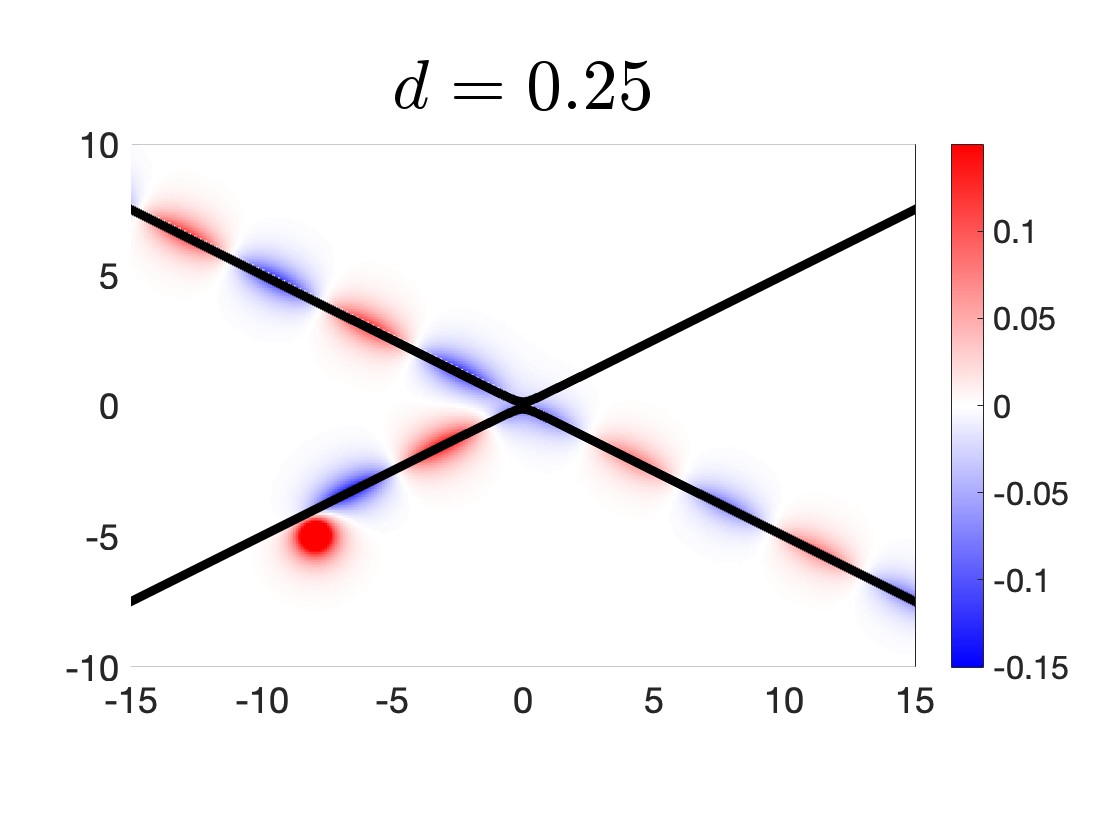}
    \includegraphics[scale=0.16]{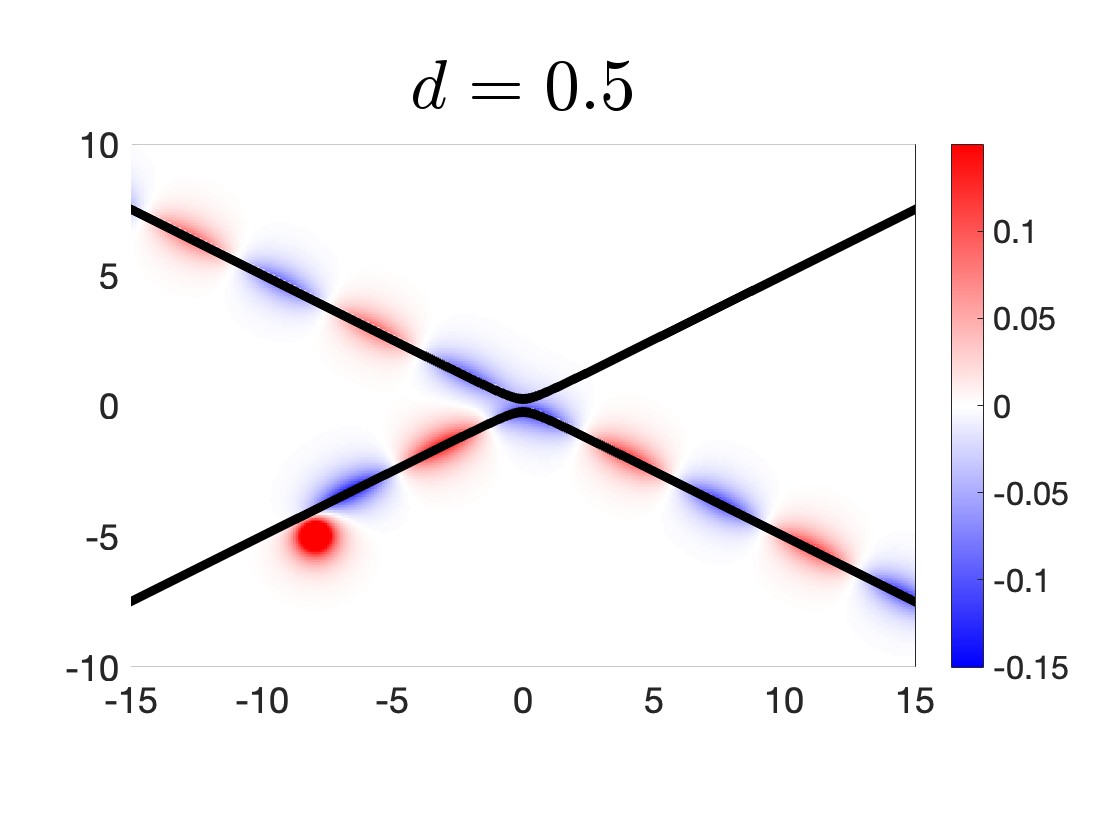}
    \includegraphics[scale=0.16]{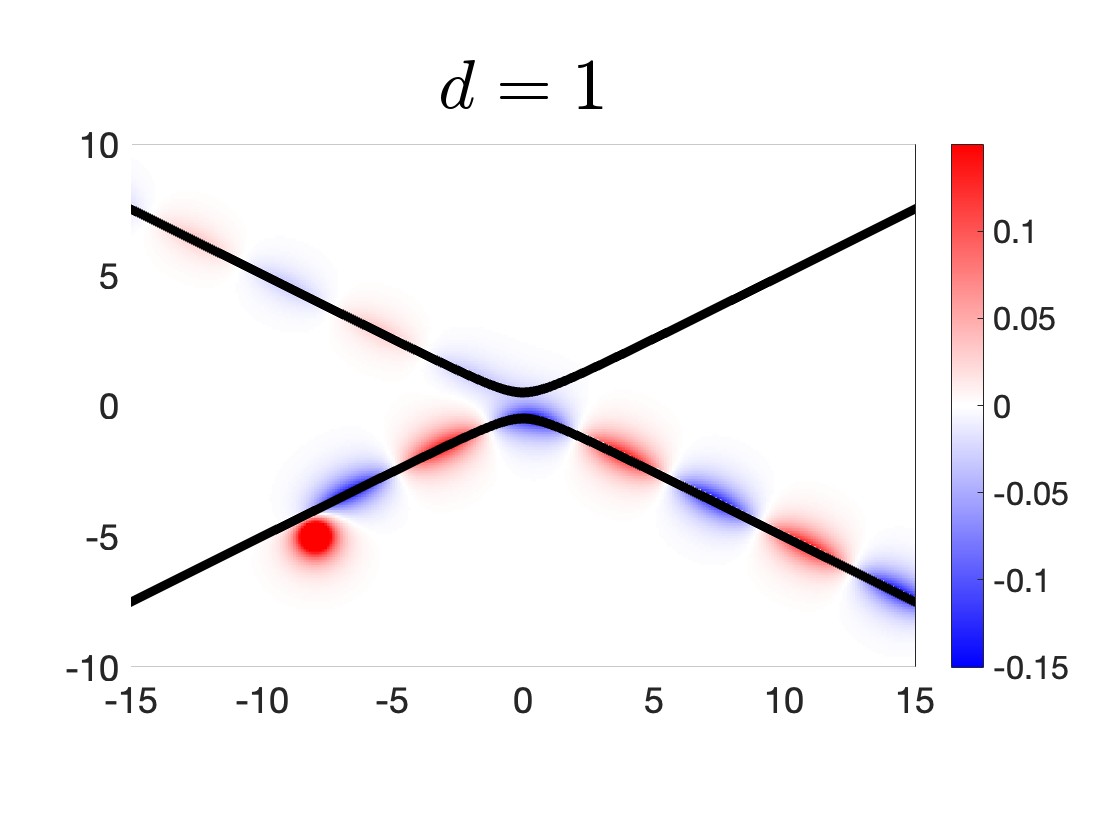}
    \includegraphics[scale=0.16]{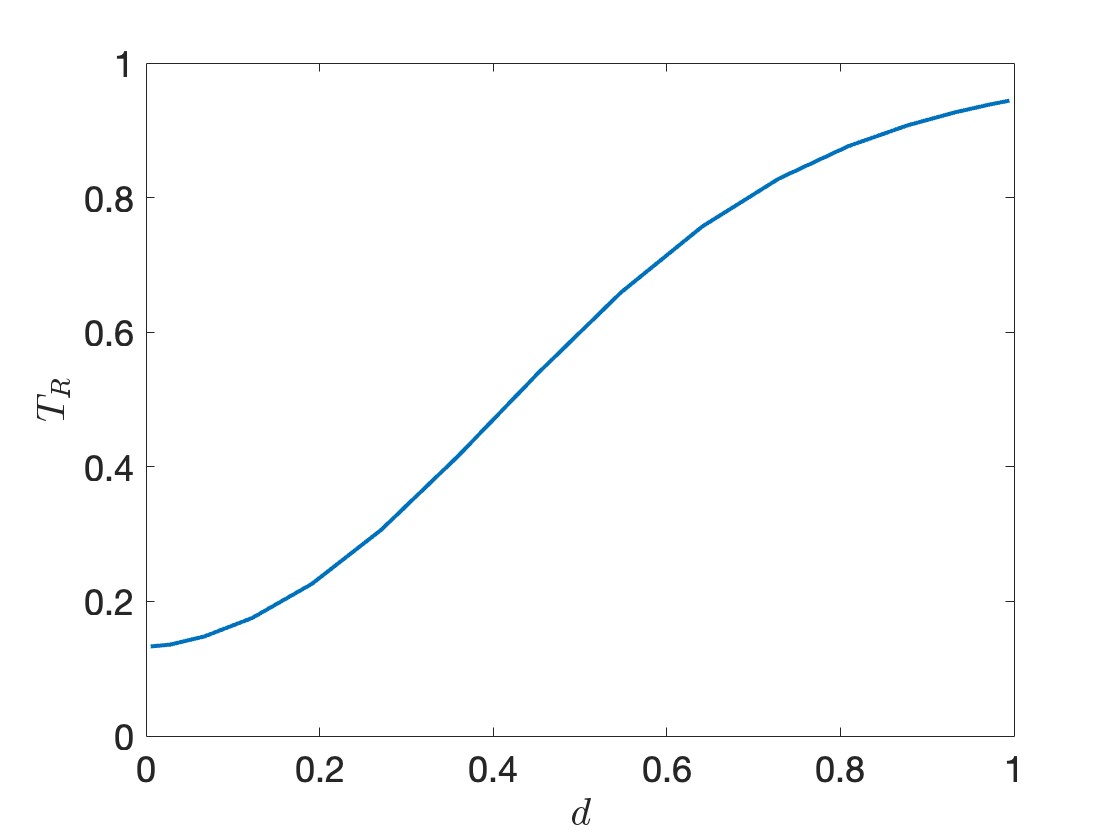}
    \caption{Same set-up as in Figure \ref{fig:scattering}, only with 
    different interfaces.}
    \label{fig:scattering_vdelta}
\end{figure}

\begin{figure}[ht!]
    \centering
    \includegraphics[scale=0.17]{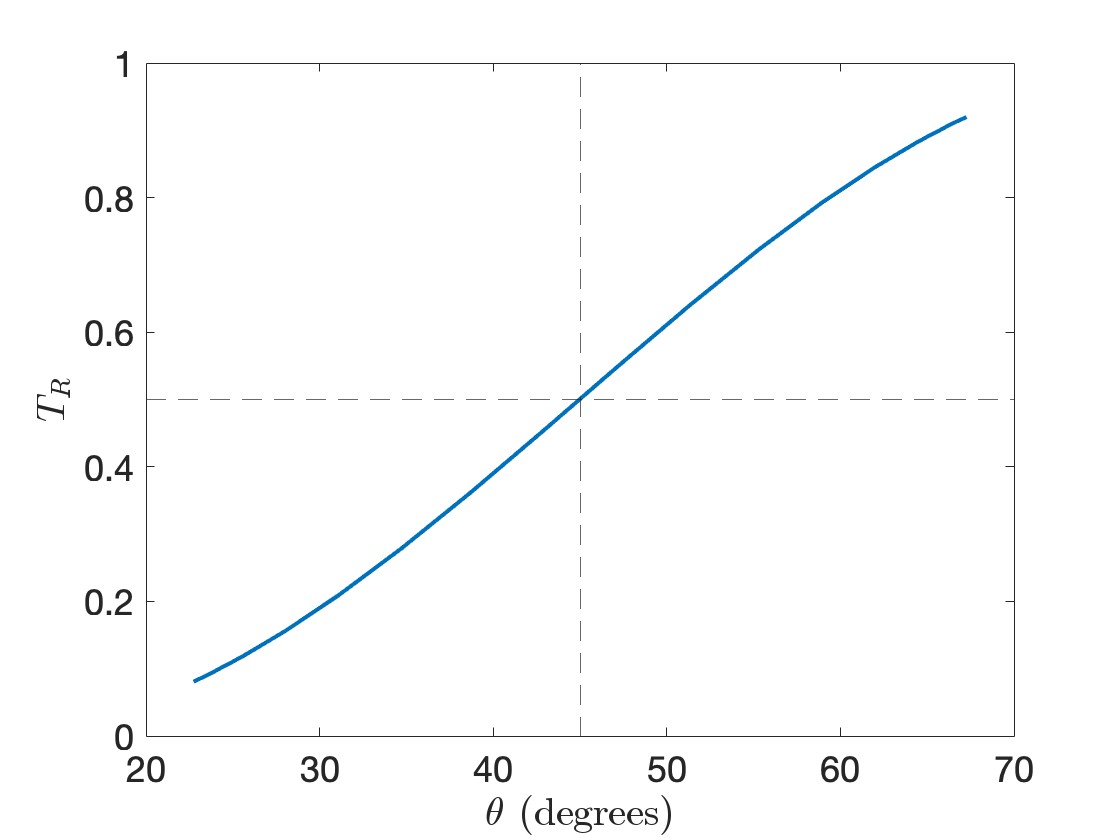}
    \caption{The $d \to 0$ limit of the transmission coefficient as a function of $\theta$, for interfaces $\Gamma_j$ parametrized by \eqref{eq:vdelta} with $a = \tan \theta$. }
    \label{fig:t_vs_theta}
\end{figure}

\section{Natural extensions and applications}\label{sec:extensions}
One could more generally consider the following Dirac equation, 
\begin{align}\label{eq:PDE_2ms}
\begin{cases}
-i \sigma_1 \partial_{x_1} u(x) -i \sigma_2 \partial_{x_2} u (x)+m_2 \sigma_3 u(x) - E u = f_2(x), \quad   & x \in \Omega_2,\\
-i \sigma_1 \partial_{x_1} u(x) -i \sigma_2 \partial_{x_2} u (x)-m_1 \sigma_3 u(x) - E u = f_1(x), \quad   & x \in \Omega_1,\\
\lim_{y\to x} u(y) = \lim_{z \to x} u(z) \quad & y \in \Omega_2, z \in \Omega_1, x \in \Gamma,
\end{cases}
\end{align}
where we assume that 
$0 < E^2 < m_1^2,m_2^2$.
This is an extension of the case $m_1= m_2$ we considered until now.
Now, we represent the solution by $u = u_i + u_s$, where
\begin{align}\label{eq:uius_2ms}
     u_i (x) := \begin{cases}
        \int_{\Omega_2} G_2 (x,y) f_2 (y) {\rm d} y, & x \in \Omega_2\\
        \int_{\Omega_1} G_1 (x,y) f_1 (y) {\rm d} y, & x \in \Omega_1
    \end{cases}
    \qquad
    u_s (x) := \begin{cases}
        \int_\mathbb{R} G_2 (x, \gamma (t)) \mu (t) {\rm d} t, & x \in \Omega_2\\
        \int_\mathbb{R} G_1 (x, \gamma (t)) \mu (t) {\rm d} t, & x \in \Omega_1
    \end{cases}
\end{align}
and
\begin{align*}
    G_j (x,y) := \frac{1}{2\pi} (-i \partial_{x_1} \sigma_1 -i \partial_{x_2} \sigma_2 + (-1)^j m_j \sigma_3 + E) K_0 (\omega_j |x-y|),
\end{align*}
with $\omega_j := \sqrt{m_j^2 - E^2}$.
We again wish to solve for $\mu$ such that $u$ is continuous across $\Gamma$. 
Recalling the shorthand $\hat{n}(t) \cdot \sigma := \hat{n}_1 (t) \sigma_1 + \hat{n}_2 (t) \sigma_2$,
standard properties of layer potentials imply that
\begin{align}\label{eq:jumpus_2ms}
\begin{split}
    [[u_s]](\gamma(t)) &= i(\hat{n}(t) \cdot \sigma) \mu (t) - i (\hat{n}(t) \cdot \sigma) (\cS'_{\omega_2} [\mu](t) - \cS'_{\omega_1} [\mu](t))\\
    &-i(\gamma'(t) \cdot \sigma) \cT'_{\omega_1, \omega_2} [\mu](t)\\
    &+\sigma_3 (m_2 \cS_{\omega_2} [\mu] (t) + m_1 \cS_{\omega_1} [\mu] (t)) + E (\cS_{\omega_2}[\mu](t) - \cS_{\omega_1} [\mu] (t)),
\end{split}
\end{align}
where the operators $\cS_\omega$, $\cS'_\omega$ and $\cT'_{\omega_1, \omega_2}$ are defined by
\begin{align*}
    \cS_\omega [\mu] (t) :&=
    \frac{1}{2\pi}\int_{\mathbb{R}} K_0 (\omega |\gamma (t)-\gamma (t')|) \mu (t') {\rm d} t',\\
    \cS'_\omega [\mu] (t):&=
    \frac{-\omega}{2\pi}\int_{\mathbb{R}}\frac{\hat{n} (t) \cdot (\gamma(t)-\gamma(t'))}{|\gamma (t)-\gamma (t')|}
    K_1 (\omega |\gamma (t)-\gamma (t')|) \mu (t') {\rm d} t',\\
    \cT'_{\omega_1, \omega_2} [\mu] (t):&=
    \frac{-1}{2\pi}\int_{\mathbb{R}} \frac{\gamma' (t) \cdot (\gamma(t)-\gamma(t'))}{|\gamma (t)-\gamma (t')|} \times \\
    &\qquad (\omega_2 K_1 (\omega_2 |\gamma (t)-\gamma (t')|) -\omega_1 K_1 (\omega_1 |\gamma (t)-\gamma (t')|)) \mu (t') {\rm d} t'.
\end{align*}
Let us first consider the flat interface $\gamma (t) := (t,0)$. 
Then $\cS'_\omega = 0$ while
\begin{align*}
    \cF_{t \to \xi} \{\cS_\omega [\mu] (t)\} = \frac{\tilde{\mu}(\xi)}{2 \sqrt{\xi^2+\omega^2}},\qquad
    \cF_{t \to \xi} \{\cT'_{\omega_1,\omega_2} [\mu] (t)\} = \frac{i\xi}{2} \left(\frac{1}{\sqrt{\xi^2+\omega_2^2}}-\frac{1}{\sqrt{\xi^2+\omega_1^2}}\right)\tilde{\mu}(\xi).
\end{align*}
Therefore, the continuity condition
$[[u_s]](\gamma (t)) + [[u_i]](\gamma (t)) =0$ implies that
\begin{align*}
    \bxi (\xi) \tilde{\mu} (\xi) = i\sigma_2 \cF_{t \to \xi}\{ [[u_i]] (\gamma (t)\},
\end{align*}
where
\begin{align*}
    \bxi (\xi) := 1-\frac{\xi}{2}\sigma_3\Big(\frac{1}{\sqrt{\xi^2+\omega_2^2}} - \frac{1}{\sqrt{\xi^2+\omega_1^2}}\Big)
    &+\frac{1}{2}\sigma_1 \Big(\frac{m_2}{\sqrt{\xi^2+\omega_2^2}}+ \frac{m_1}{\sqrt{\xi^2+\omega_1^2}}\Big)\\
    &- i\frac{E}{2} \sigma_2\Big(\frac{1}{\sqrt{\xi^2+\omega_2^2}} - \frac{1}{\sqrt{\xi^2+\omega_1^2}}\Big).
\end{align*}
To obtain the above, we applied \eqref{eq:jumpus_2ms}, multiplied both sides of $[[u_s]](\gamma (t)) = -[[u_i]](\gamma (t))$ by $-i \hat{n} (t) \cdot \sigma = -i\sigma_2$, and used the property $\sigma_i \sigma_j + \sigma_j \sigma_i = 2\delta_{ij}$.
We observe that $\bxi (\xi)$ is singular when $\xi = \pm E$, with (assuming $0 < m_1, m_2$ for concreteness)
\begin{align*}
    \bxi (-E) = \begin{pmatrix}
        1+\alpha & 1-\alpha\\
        1+\alpha & 1-\alpha
    \end{pmatrix},
    \qquad
    \bxi (E) = \begin{pmatrix}
        1-\alpha & 1-\alpha\\
        1+\alpha & 1+\alpha
    \end{pmatrix},
    \qquad \alpha := \frac{E}{2} \Big(\frac{1}{m_2} - \frac{1}{m_1}\Big).
\end{align*}
We again must invert these singularities. In contrast with the $m_1 = m_2$ case, the singular matrices $\bxi (-E)$ and $\bxi (E)$ are not the same, and thus must be preconditioned separately. We propose to multiply $\bxi (\xi)$ by the function $1 + i(m_1^2+m_2^2) (\mmat_L \tilde{R}_L (\xi) - \mmat_R \tilde{R}_R (\xi))$, where
\begin{align*}
    \mmat_L :&= \begin{pmatrix}
        1-\alpha & -(1-\alpha)\\
        -(1+\alpha) & 1+\alpha
    \end{pmatrix},
    \qquad
    \mmat_R := \begin{pmatrix}1-\alpha & -(1+\alpha)\\
    -(1-\alpha) & 1+\alpha
    \end{pmatrix},\\
    \tilde{R}_L (\xi) :&= \frac{1}{2E (\xi+E)}, \qquad \tilde{R}_R (\xi) := \frac{1}{2E (\xi - E)}.
\end{align*}
That is, we will solve
\begin{align}\label{eq:solverho_2ms_ft}
    \bxi (\xi) (1 + i(m_1^2+m_2^2) (\mmat_L \tilde{R}_L (\xi) - \mmat_R \tilde{R}_R (\xi))) \tilde{\rho} (\xi) = i\sigma_2 \cF_{t \to \xi}\{ [[u_i]] (\gamma (t))\}
\end{align}
for $\tilde{\rho}$, where the matrix-valued function multiplying $\tilde{\rho}$ on the above left-hand side is smooth in $\xi$ with eigenvalues bounded away from zero.

Observe that the inverse Fourier transforms of $\tilde{R}_{L,R}$ are
\begin{align*}
    \check{R}_L (t) := \frac{1}{2iE} H(-t) e^{-iEt}, \qquad \check{R}_R (t) := \frac{-1}{2iE}H(t) e^{iEt},
\end{align*}
with $H$ the Heaviside step function. 
Therefore, in real space, \eqref{eq:solverho_2ms_ft} becomes
\begin{align}\label{eq:solverho_2ms_flat}
\begin{split}
    (1 +i \sigma_3 \cT'_{\omega_1, \omega_2}
    &+\sigma_1 (m_2 \cS_{\omega_2} + m_1 \cS_{\omega_1})
    -i\sigma_2 E (\cS_{\omega_2} - \cS_{\omega_1}))\\
    (1&+i(m_1^2+m_2^2)(M_L \cR_L+M_R \cR_R)) [\rho] (t) = i\sigma_2 [[u_i]] (\gamma (t)),
\end{split}
\end{align}
where the operators $\cR_{L,R}$ are defined by
\begin{align*}
    \cR_{L} [\rho] (t) :&= \frac{1}{2iE}\int_{\mathbb{R}} H(-(t-t')) e^{-iE(t-t')} \rho(t') {\rm d} t' = \frac{1}{2iE}\int_{t}^{\infty} e^{-iE(t-t')} \rho(t') {\rm d} t',\\
    \cR_{R} [\rho] (t) :&= \frac{1}{2iE}\int_{\mathbb{R}} H(t-t') e^{iE(t-t')} \rho(t') {\rm d} t' = \frac{1}{2iE}\int_{-\infty}^t e^{iE(t-t')} \rho(t') {\rm d} t'.
\end{align*}
The equation \eqref{eq:solverho_2ms_flat} is for the flat interface case. More generally,
the above derivation motivates the following integral equation for $\rho$,
\begin{align}\label{eq:solverho_2ms}
    \begin{split}
        (1 &- (\cS'_{\omega_2} - \cS'_{\omega_1})
        +i \sigma_3 \cT'_{\omega_1, \omega_2}
        -i (\hat{n} \cdot \sigma)\sigma_3 (m_2 \cS_{\omega_2} + m_1 \cS_{\omega_1})\\
        &-i(\hat{n} \cdot \sigma)E (\cS_{\omega_2} - \cS_{\omega_1}))
        (1+i(m_1^2+m_2^2)(M_L \cR_L+M_R \cR_R)) [\rho] (t) \\
        &= i\hat{n} \cdot \sigma [[u_i]] (\gamma (t)),
    \end{split}
\end{align}
where the first factor on the left-hand side is obtained by multiplying \eqref{eq:jumpus_2ms} by $-i \hat{n} \cdot \sigma$.
Once the solution $\rho$ of \eqref{eq:solverho_2ms} is obtained, we then set
\begin{align*}
    \mu (t) := (1+i(m_1^2+m_2^2)(\mmat_L \cR_L+\mmat_R \cR_R)) [\rho] (t).
\end{align*}
The function $u:= u_i + u_s$ then satisfies the PDE \eqref{eq:PDE_2ms}, where $u_i$ and $u_s$ are defined in \eqref{eq:uius_2ms} with the latter depending on $\mu$.

\appendix

\section{Proofs of main analytical results}\label{sec:proofs}
This section contains proofs of the results in Section \ref{sec:analytical}. 
We will be working with the weighted $L^2$ spaces that were defined in Section \ref{subsec:bie} as follows. 
For $\alpha\in \mathbb{R}$, define $w_\alpha (t) := e^{\alpha|t|}$ and $L^2_\alpha := \{\rho \in L^2 (\mathbb{R}) : w_\alpha \rho \in L^2(\mathbb{R})\}$.
Throughout this section, if $A :\mathcal{H}_1 \rightarrow \mathcal{H}_2$ is a bounded linear operator, we denote its operator norm by $\norm{A} := \sup_{\rho \in \mathcal{H}_1, \norm{\rho}_{\mathcal{H}_1} =1} \norm{A \rho}_{\mathcal{H}_2}$, where $\norm{\cdot}_{\mathcal{H}_j}$ is the norm on $\mathcal{H}_j$.

Before directly proving our main results, we collect useful properties of various integral operators related to $\cL$ and $\cP$ with Lemmas \ref{lemma:invLP} and \ref{lemma:MHS} below.
Recall that $\cL$ and $\cP$ are defined by
\begin{align*}
    \cL [\rho] (t) :&= \rho(t) -2m \cobmat^* (t) (i\hat{n}(t) \cdot \sigma) \sigma_3 
    \frac{1}{2\pi} \int_{\mathbb{R}} K_0 (\omega |\gamma (t)-\gamma (t')|) \cobmat (t') \rho (t') {\rm d} t',\\
    \cP [\rho] (t) :&= \rho (t) + \mmat \frac{m^2}E \int_{\Rm} e^{iE|t-t'|} \rho (t'){\rm d}t',
\end{align*}
where $\mmat$ is defined in~\eqref{eqn:mdef}.

For ease of notation, we will define the operator $Q$ by
\begin{align}\label{eq:Q}
    Q\rho (t) := \frac{m^2}E \int_{\Rm} e^{iE|t-t'|} \rho (t'){\rm d}t', \qquad \rho \in L^2 (\mathbb{R}),
\end{align}
so that $\cP = 1 + \mmat Q$.
In the following proofs, we will show that the equation 
\begin{align}\label{eq:solverho_pfs}
    \cL\cP [\rho] (t) =i \cobmat^* \hat{n} \cdot \sigma [[u_i]] (\gamma (t))
\end{align}
that we wish to solve can be written as a compact perturbation of the corresponding flat-interface problem.
To this end, define the operators 
\begin{align}\label{eqn:flatl}
\flatL := 1+2m\sigma_3 \cSflat
\end{align}
and $\cSflat$ by
\begin{align*}
    \cSflat [\mu] (t) := \frac{1}{2\pi} \int_{\mathbb{R}} K_0 (\omega |t-t'|) \mu (t') {\rm d} t',
\end{align*}
so that $\flatL$ and $\cSflat$ are the operators $\flatL$ and $\cS_\omega$ in the case of a flat interface.
We then have
\begin{lemma}\label{lemma:invLP}
Fix $m_0, E_0 \in \mathbb{R}$ such that $0 < |E_0| < |m_0|$, and set 
$m=\lambda m_0$ and $E = \lambda E_0$ for $\lambda \in \mathbb{R}$.
For all 
$\alpha>0$ sufficiently small, 
the operator $(\flatL \flatP)^{-1}$ is bounded on $L^2_\alpha$ with $\norm{(\flatL \flatP)^{-1}} \le C$ uniformly in $\lambda \in [1,\infty)$.
\end{lemma}
\begin{proof}
    Assume $m_0>0$ for concreteness.
    Observe that $\flatL$ and $\cP$ are both diagonal matrix-valued operators, and thus so is $\flatL \cP$.
    By \cite[Lemma 5]{bal2022integral}, the operator $(\flatL \cP)_{22} = (1-2m \cSflat) (1+Q)$ has an inverse which is bounded on $L^2_\alpha$ uniformly in $\lambda \in [1, \infty)$, whenever $\alpha$ is sufficiently small.
    The other nonzero entry, $(\flatL \cP)_{11} = 1+2m \cSflat$, admits the Fourier representation $$1 + \frac{m}{\sqrt{\xi^2+\omega^2}} = 1 + \frac{\lambda m_0}{\sqrt{\xi^2+\lambda^2 \omega_0^2}}, \qquad \omega_0 := \sqrt{m_0^2 - E_0^2},$$ which is bounded away from zero uniformly in $\xi \in \mathbb{R}$ and $\lambda \in [1, \infty)$. 
\end{proof}

Lemma \ref{lemma:invLP} allows for the following reformulation of our integral equation. Writing $\cL \cP = \flatL \cP + (\cL - \flatL ) \cP$ and multiplying both sides of 
\eqref{eq:solverho_pfs} by $(\flatL \flatP)^{-1}$, we obtain the equivalent problem
\begin{align}\label{eq:eqM}
    (1 + \cM) [\rho] (t) = i(\flatL \flatP)^{-1}\cobmat^* \hat{n} \cdot \sigma [[u_i]] (\gamma (t)),
\end{align}
where $\cM := (\flatL \flatP)^{-1}(\cL - \flatL ) \cP$.
We now show that $\cM$ is Hilbert-Schmidt on $L^2_\alpha$, with the associated (Hilbert-Schmidt) norm going to zero as $m$ and $E$ go to infinity.
In the following, we let $\norm{\cdot}_2$ denote the Hilbert-Schmidt norm on $L^2_\alpha$ and $\vertiii{\cdot}_{2}$ the Hilbert-Schmidt norm on $L^2$.
\begin{lemma}\label{lemma:MHS}
    Fix $m_0, E_0 \in \mathbb{R}$ such that $0 < |E_0| < |m_0|$, and set 
    $m=\lambda m_0$ and $E = \lambda E_0$ for $\lambda \in \mathbb{R}$.
    For all $\alpha > 0$ sufficiently small and $\delta > 0$, the operator $\cM$ is Hilbert-Schmidt on $L^2_\alpha$ whenever $\lambda \in [1,\infty)$, with
    $\norm{\cM}_2 \le C_{\alpha,\delta}\lambda^{-\frac{1}{2} +\delta}$. 
\end{lemma}
\begin{proof}
    Assume $m_0>0$ for concreteness.
    We write
    \begin{align*}
        \cL - \flatL = 2m\sigma_3 (\cS_\omega - \cSflat) - 2im \cobmat^* (\hat{n} \cdot \sigma) \sigma_3 [\cS_\omega, \cobmat] =: \ldiff_1 + \ldiff_2.
    \end{align*}
    Define $\cM_j := (\flatL \cP)^{-1}\ldiff_j \cP$ so that $\cM = \cM_1 + \cM_2$.
    Recalling the proof of Lemma \ref{lemma:invLP} above,
    note that $(\flatL \cP)^{-1}, \ldiff_1$ and $\cP$ are all diagonal matrix-valued operators, hence so is $\cM_1$.
    The second diagonal entry of $\cM_1$ is
    \begin{align*}
        (\cM_1)_{22} = -[(1-2m \cSflat) (1+Q)]^{-1}2m (\cS_\omega - \cSflat) (1+Q),
    \end{align*}
    which satisfies
    $\norm{(\cM_1)_{22}}_2 \le C/\sqrt{\lambda}$ uniformly in $\lambda \ge 1$ for all $\alpha>0$ sufficiently small, by \cite[Proposition 7]{bal2022integral}. 
    The first diagonal entry is
    \begin{align*}
        (\cM_1)_{11} = [(\flatL \cP)_{11}]^{-1}2m (\cS_\omega - \cSflat),
    \end{align*}
    with $[(\flatL \cP)_{11}]^{-1}$ uniformly bounded in $L^2_\alpha$ by Lemma \ref{lemma:invLP}.
    In \cite[equation (67)]{bal2022integral}, it was shown that $\cS_\omega - \cSflat$ admits a 
    Schwartz kernel 
    $q_\sharp$ that satisfies
    \begin{align}\label{eq:qsharp}
        0 \le q_\sharp (t,t') \le C \omega |t-t'|^3 e^{-\dc \chi (t,t')} K_1 (c \omega |t-t'|),
    \end{align}
    where
    \begin{align}\label{eq:chi}
        \chi (t,t') := \begin{cases}
            \min \{|t|,|t'|\}, & tt' > 0\\
            0, & \text{else}
        \end{cases}
    \end{align}
    and $K_1$ is the modified Bessel function of the second kind.
    Define $W_\alpha : L^2_\alpha \rightarrow L^2$ by
    \begin{align*}
        W_\alpha [\rho] (t) := w_\alpha (t) \rho (t),
    \end{align*}
    so that $W_\alpha$ is the operator of point-wise multiplication with the function $w_\alpha (t) = e^{\alpha |t|}$.
    Write $m (\cS_\omega - \cSflat) =  W_\alpha^{-1} [W_\alpha m (\cS_\omega - \cSflat) W_\alpha^{-1}] W_\alpha$ and note that $W_\alpha$ and $W_\alpha^{-1}$ are bounded (and independent of $\lambda$). Thus, to obtain the desired bound on $(\cM_1)_{11}$, it suffices to show that $\cS_\Delta := W_\alpha m (\cS_\omega - \cSflat) W_\alpha^{-1}$ is Hilbert-Schmidt on $L^2$ with $\vertiii{\cS_\Delta}_2 \le C/\sqrt{\lambda}$ (uniformly in $\lambda$ for $\alpha$ sufficiently small).
    By \eqref{eq:qsharp}, the Schwartz kernel $q_\Delta$ of $\cS_\Delta$ satisfies
    \begin{align*}
        |q_\Delta (t,t')| \le C e^{\alpha |t|} \lambda^2 |t-t'|^3 e^{-\dc \chi (t,t')} K_1 (c \lambda \omega_0 |t-t'|)e^{-\alpha |t'|},
    \end{align*}
    where $\omega_0 := \sqrt{m_0^2 - E_0^2}$.
    Using \cite[proof of Proposition 7]{bal2022integral}, we conclude that $\norm{q_\Delta}_{L^2 (\mathbb{R}^2)} \le C \lambda^{-3/2}$, hence $\norm{(\cM_1)_{11}}_2 \le C \lambda^{-3/2}$ uniformly in $\lambda \ge 1$ for any $\alpha > 0$ sufficiently small.
    Combining our bounds on the diagonal entries of $\cM_1$, we have shown that $\norm{\cM_1}_2 \le C/\sqrt{\lambda}$
    uniformly in $\lambda \ge 1$ for all $\alpha > 0$ sufficiently small.
    
    It remains to establish a similar bound for $\cM_2$.
    As in \cite[proofs of Lemma 6 and Proposition 7]{bal2022integral}, our strategy is to write
    \begin{align*}
        \cM_2 = (\flatL \cP)^{-1} A_1 A_2 A_3,
    \end{align*}
    where 
    the operators $A_j$ are given by
    \begin{align*}
        A_1 := W_\alpha^{-1}, \qquad A_2 := W_\alpha \Delta_2 W_\alpha, \qquad A_3 := W_\alpha^{-1} \cP.
    \end{align*}
    It is clear that $A_1 : L^2 \rightarrow L^2_\alpha$ is bounded and independent of $\lambda$.
    Since $$|Q \rho| (t) \le \frac{m^2}{E} \norm{\rho}_{L^1} \le C \lambda \norm{\rho}_{L^2_\alpha},$$
    it follows that $A_3 :L^2_\alpha \rightarrow L^2$ is bounded with $\norm{A_3} \le C \lambda$ uniformly in $\lambda \ge 1$.
    Since $(\flatL \cP)^{-1}$ is bounded on $L^2_\alpha$ uniformly in $\lambda$ by Lemma \ref{lemma:invLP}, it suffices to show that $A_2$ is Hilbert-Schmidt on $L^2$ with $\vertiii{A_2}_2\le C \lambda^{-\frac{3}{2}+\delta}$ uniformly in $\lambda \ge 1$ whenever $\alpha > 0$ is sufficiently small.

    Observe that the Schwartz kernel of $[\cS_\omega, \cobmat]$ is given by
    \begin{align*}
        (t, t') \mapsto \frac{1}{2\pi} K_0 (\omega |\gamma (t) - \gamma (t')|) (\cobmat (t') - \cobmat (t)).
    \end{align*}
    The assumption \eqref{eq:beta} on $\Gamma$ at infinity implies that
    \begin{align}\label{eq:diff1}
        |\cobmat (t') - \cobmat (t)| \le C e^{-\delta^{-1} \dc \chi (t,t')} 
    \end{align}
    for some constant $\dc > 0$ (recall the definition of $\chi$ in \eqref{eq:chi}),
    while
    smoothness of $\Gamma$ implies that
    \begin{align}\label{eq:diff2}
        |\cobmat (t') - \cobmat (t)| \le C |t'-t|.
    \end{align}
    Interpolating between \eqref{eq:diff1} and \eqref{eq:diff2} yields
    \begin{align}\label{eq:interp}
        |\cobmat (t') - \cobmat (t)| \le C e^{-\dc \chi (t,t')} |t'-t|^{1-\delta}.
    \end{align}
    Since $2im \cobmat^* (\hat{n} \cdot \sigma) \sigma_3$ is uniformly bounded,
    we conclude that the Schwartz kernel $q$ of $A_2 =W_\alpha \Delta_2 W_\alpha$ satisfies
    \begin{align}\label{eq:qlambda}
        |q(t,t')| \le C w_\alpha (t) w_\alpha (t') |t-t'|^{1-\delta} K_0 (c \omega_0 \lambda |t-t'|) e^{-\dc \chi (t,t')}.
    \end{align}
    It follows that
    \begin{align*}
        \vertiii{A_2}_2^2 &= \norm{q}^2_{L^2 (\mathbb{R}^2)} \\
        &\le 
        C\int_{\mathbb{R}^2}
        w_\alpha^2 (t) w_\alpha^2 (t') |t-t'|^{2(1-\delta)} K_0^2 (c \omega_0 \lambda |t-t'|) e^{-2\dc \chi (t,t')} {\rm d} t {\rm d}t'\\
        &=\frac{C}{2}\int_{\mathbb{R}^2}
        w_\alpha^2 \Big(\frac{\xi+\zeta}{2}\Big) w_\alpha^2 \Big(\frac{-\xi+\zeta}{2}\Big) |\xi|^{2(1-\delta)} K_0^2 (c \omega_0 \lambda |\xi|) e^{-2\dc \chi (\frac{\xi+\zeta}{2}, \frac{-\xi+\zeta}{2})} {\rm d} \xi {\rm d}\zeta.
    \end{align*}
    Using that
    $w_\alpha^2 \Big(\frac{\xi+\zeta}{2}\Big) w_\alpha^2 \Big(\frac{-\xi+\zeta}{2}\Big) = \exp \{ \alpha (|\xi+\zeta| + |-\xi + \zeta|)\}$
    and $|\xi+\zeta| + |-\xi + \zeta| \le |\lambda \xi+\zeta| + |-\lambda \xi + \zeta|$ for all $\lambda \ge 1$ and $(\xi,\zeta) \in \mathbb{R}^2$,
    we conclude that
    $$w_\alpha^2 \Big(\frac{\xi+\zeta}{2}\Big) w_\alpha^2 \Big(\frac{-\xi+\zeta}{2}\Big) \le w_\alpha^2 \Big(\frac{\lambda \xi+\zeta}{2}\Big) w_\alpha^2 \Big(\frac{-\lambda \xi+\zeta}{2}\Big).$$
    Moreover, the definition \eqref{eq:chi} of $\chi$ directly implies that
    \begin{align*}
        \chi\Big(\frac{\lambda \xi+\zeta}{2}, \frac{-\lambda \xi+\zeta}{2}\Big) \le 
        \chi \Big(\frac{\xi+\zeta}{2}, \frac{-\xi+\zeta}{2}\Big)
    \end{align*}
    for all $\lambda \ge 1$ and $(\xi,\zeta) \in \mathbb{R}^2$.
    Therefore,
    \begin{align*}
        \vertiii{A_2}_2^2 &\le
        C\int_{\mathbb{R}^2}
        w_\alpha^2 \Big(\frac{\lambda \xi+\zeta}{2}\Big) w_\alpha^2 \Big(\frac{-\lambda \xi+\zeta}{2}\Big) |\xi|^{2(1-\delta)} K_0^2 (c \omega_0 \lambda |\xi|) e^{-2\dc \chi (\frac{\lambda \xi+\zeta}{2}, \frac{-\lambda \xi+\zeta}{2})} {\rm d} \xi {\rm d}\zeta\\
        &=C\lambda^{-2(1-\delta)-1}\int_{\mathbb{R}^2}
        w_\alpha^2 \Big(\frac{\xi+\zeta}{2}\Big) w_\alpha^2 \Big(\frac{-\xi+\zeta}{2}\Big) |\xi|^{2(1-\delta)} K_0^2 (c \omega_0 |\xi|) e^{-2\dc \chi (\frac{\xi+\zeta}{2}, \frac{-\xi+\zeta}{2})} {\rm d} \xi {\rm d}\zeta\\
        &=C\lambda^{-2(1-\delta)-1},
    \end{align*}
    where the last equality (redefines the constant $C$ for ease of notation and) requires that $\alpha>0$ is sufficiently small so that the $\lambda$-independent integral is finite.
    We conclude that $\vertiii{A_2}_2 \le C \lambda^{-\frac{3}{2}+\delta}$ uniformly in $\lambda \ge 1$ for any $\alpha > 0$ sufficiently small, and the proof is complete.
\end{proof}
Using the above lemmas, we can now prove one of our main results.
\begin{proof}[Proof of Theorem \ref{thm:invL}]
Let $\alpha > 0$ be as small as necessary, 
and fix $0 < \delta < \min\{\alpha, 1/2\}$. 
Define $\cM := (\flatL \flatP)^{-1}(\cL - \flatL ) \cP$ as in Lemma \ref{lemma:MHS}, and $Z_\delta := [1,\infty) \times (-\delta, \delta)$.
Since $\delta < 1/2$, 
we know that $\Re (\lambda^2) > 3/4$ whenever $\lambda \in Z_\delta$.
Hence, the modified Bessel function of the second kind, $K_0 (\lambda r)$, is exponentially decaying as $r \rightarrow +\infty$ (at a rate that is uniform in $\lambda$) and so $(\flatL \flatP)^{-1}$ and $\cL - \flatL$ are holomorphic in $\lambda \in Z_\delta$.
Since $\delta < \alpha$, Lemma \ref{lemma:MHS} implies that
$Z_\delta \ni \lambda \mapsto \cM$ defines a holomorphic compact-operator-valued function,
with all eigenvalues of $\cM$ going to zero as $\Re \lambda \rightarrow \infty$. 
It then follows from analytic perturbation theory \cite[Theorem VII.1.9]{kato2013perturbation} that $\cM$ has an eigenvalue of $-1$ for at most a finite number of $\lambda \in [1,\infty)$. By the equivalence of the integral equations \eqref{eq:solverho_pfs} and \eqref{eq:eqM}, the proof is complete.
\end{proof}

When the right-hand side of~\eqref{eq:solverho} is $m$-times differentiable, we can use the previous result to establish improved regularity of the density $\rho.$

\begin{lemma}\label{lem_boot}
    Suppose that $0<\alpha<\omega$ and $\rho \in L^2_\alpha(\mathbb{R};\mathbb{C}^2)$ is a solution to~\eqref{eq:solverho} and moreover that the right-hand side is in $H^m_\alpha(\mathbb{R};\mathbb{C}^2).$ Additionally, for $q>\alpha$ suppose that $\gamma$ (an arclength parameterization of $\Gamma$) is in $C^{k,q}(\mathbb{R};\mathbb{R}^2),$ for any $k \ge 0,$ where $C^{k,q}$ is the space of functions defined by
    $$C^{k,q} := \{ \gamma \in C^k(\mathbb{R};\mathbb{R}^2)\,\,|\,\, e^{q|t|}\gamma^{(j)}(t) \in C(\mathbb{R};\mathbb{R}^2),\, j=0,\dots,k\}.$$
    Then $\rho \in H^m_\alpha(\mathbb{R};\mathbb{C}^2).$ 
\end{lemma}
\begin{proof}
    The proof follows from a standard bootstrap argument. We first introduce some convenient notation. Recalling the definitions of the operators $\mathcal{L},$ $\mathcal{P}$ and $\mathcal{L}_0$  defined in \eqref{eq:int},\eqref{eq:cP}, and \eqref{eqn:flatl}, respectively, then $\rho \in L^2_\alpha(\mathbb{R};\mathbb{C}^2)$ is a solution to \eqref{eq:int} provided that
    $$\mathcal{L}\mathcal{P}[\rho] =f,$$
    where $f = i \cobmat^* \hat{n} \cdot \sigma [[u_i]] \circ \gamma$ is the right-hand side of \eqref{eq:int}.

    It follows immediately that
    \begin{align}\label{eq:rho_boot}
    \rho = (I-\mathcal{L}_0 \mathcal{P})[\rho] + (\mathcal{L}_0-\mathcal{L})\mathcal{P}[\rho] + f.
    \end{align}

    Next, we observe from \eqref{eq:a1} and \eqref{eq:a2} that in the Fourier domain $I-\mathcal{L}_0\mathcal{P}$ is multiplication by
    $$\left(1+\frac{m}{\sqrt{\xi^2+\omega^2}}\sigma_3\right) \left(1 -\mmat \frac{2im^2}{\xi^2-E^2} \right) - 1$$
    which simplifies to 
    $$-\mmat \left[\frac{2im^2}{\xi^2-E^2}\left( \frac{\sqrt{\xi^2+m^2-E^2}-|m|}{\sqrt{\xi^2+\omega^2}}\right)+\frac{|m|}{\sqrt{\xi^2+\omega^2}} \right] + (1-\mmat) \frac{|m|}{\sqrt{\xi^2+\omega^2}}, $$
    where we recall that $\omega^2 = m^2-E^2.$ This function is analytic in the strip
    $$\{\xi \in \mathbb{C}\,|\, |\Im \,(\xi)| < \omega\}. $$
    Moreover, as $\xi \to \pm \infty$ near the real axis, it is bounded by $1/|\xi|$ in each component. It follows immediately from the analyticity and the decay that 
    $$\mathcal{L}_0 \mathcal{P}: H^{s}_\alpha(\mathbb{R};\mathbb{C}^2) \to H^{s+1}_\alpha(\mathbb{R}; \mathbb{C}^2),$$ 
    boundedly for any $s \ge 0.$ Here $H^s_\alpha(\mathbb{R};\mathbb{C}^2)$ denote the weighted Sobolev spaces with weight function $w_{\alpha}(x) = e^{\alpha |x|}$ with $H^0_\alpha(\mathbb{R};\mathbb{C}^2) = L^2_\alpha(\mathbb{R};\mathbb{C}^2).$

    Next we turn to $(\mathcal{L}_0-\mathcal{L})\mathcal{P}.$ We note that for any $s \in \mathbb{N}\cup\{0\}$ $$\mathcal{P}: H^s_\alpha(\mathbb{R}; \mathbb{C}^2) \to C^{s+1}(\mathbb{R};\mathbb{C}^2),$$
    boundedly. Moreover, $(\mathcal{L}_0-\mathcal{L})$ is an integral operator with a kernel $k(s,t)$ which is $O(|s-t|^4\log|t-s|) +C^{m}$ for small $|s-t|$ and decays exponentially for large $|s|+|t|.$ We emphasize that this is much stronger as a decay condition than $\mathcal{L}$ and $\mathcal{L}_0$ which only decay exponentially for large $|s-t|.$ It follows from standard arguments that $(\mathcal{L}_0-\mathcal{L})$ is bounded as a map from $C^{s+1}(\mathbb{R};\mathbb{C}^2)$ to $H^{s+2}_\alpha(\mathbb{R};\mathbb{C}^2).$

    Thus, returning to \eqref{eq:rho_boot}, we see that
    $$\rho \in L^2_\alpha(\mathbb{R};\mathbb{C}^2) \cap \left[H_\alpha^1(\mathbb{R};\mathbb{C}^2) +H_\alpha^2(\mathbb{R};\mathbb{C}^2) + H^m_\alpha(\mathbb{R};\mathbb{C}^2) \right]$$
    and hence $\rho \in H^1_\alpha(\mathbb{R};\mathbb{C}^2).$ Again using \eqref{eq:rho_boot} we see that $\rho \in H^1_\alpha(\mathbb{R};\mathbb{C}^2)$ implies that $\rho \in H^2_\alpha(\mathbb{R};\mathbb{C}^2).$ We can continue in this way up to $m$, whence it follows that
    $\rho \in H^m_\alpha(\mathbb{R};\mathbb{C}^2)$ as required.
\end{proof}

Next, we show that a solution of our integral formulation \eqref{eq:solverho_pfs} generates a solution for the Dirac equation \eqref{eq:PDE}.
\begin{proof}[Proof of Theorem \ref{thm:usol}]
We begin by verifying the first two lines of \eqref{eq:PDE}. For $x \in \Omega_j$, it follows that
\begin{align*}
    (-i \partial_{x_1} \sigma_1- i \partial_{x_2} \sigma_2 + (-1)^j m \sigma_3 - E)
    (-i \partial_{x_1} \sigma_1- i \partial_{x_2} \sigma_2 + (-1)^j m \sigma_3 + E)
    =-\Delta + \omega^2,
\end{align*}
hence (by dominated Lebesgue and the fact that $\frac{1}{2\pi}K_0 (\omega|x-y|)$ is the Green's function for the operator $-\Delta + \omega^2$)
\begin{align*}
    (-i \partial_{x_1} \sigma_1- i \partial_{x_2} \sigma_2 + (-1)^j m \sigma_3 - E) u_i (x) &= \frac{1}{2\pi} \int_{\Omega_j} (-\Delta_x+\omega^2) K_0 (\omega|x-y|) f_j (y) {\rm d} y\\
    &= f_j (x)
\end{align*}
and
\begin{align*}
    (-i \partial_{x_1} \sigma_1- i \partial_{x_2} \sigma_2 + (-1)^j m \sigma_3 - E) u_s (x) &= \frac{1}{2\pi} \int_{\Gamma} (-\Delta_x+\omega^2) K_0 (\omega|x-\gamma (t)|) \mu (t) {\rm d} t\\
    &= 0.
\end{align*}
It follows that $u = u_i +u_s$ satisfies the first two lines of \eqref{eq:PDE}. To show that $u$ satisfies the third line of~\eqref{eq:PDE}, i.e. is continuous across the interface $\Gamma$, we use the smoothness of the data $[[u_i]]$, together with Lemma \ref{lem_boot} to argue that $\rho \in L^2_\alpha(\mathbb{R};\mathbb{C}^2) \cap C^\infty(\mathbb{R};\mathbb{C}^2)$ and therefore that
$\mu = \cobmat \cP \rho\in C^\infty(\mathbb{R};\mathbb{C}^2)$. It follows by standard results in integral equations \cite{colton2013integral} that the limits in~\eqref{eq:us} as $x\to \Gamma$ in $\Omega_1$ and $\Omega_2$ exist pointwise. Continuity then follows by construction from the integral equation.

It now remains to verify \eqref{eq:out0}.
Since $G_j (x,y)$ decays exponentially in $|x-y|$,
the Lebesgue dominated convergence theorem implies that
the second line of \eqref{eq:out0} holds.
For the first line of \eqref{eq:out0}, we first observe that $u_i$ and all its derivatives decay rapidly at infinity.
Moreover, since $G_j (\gamma (t) + r \hat{n} (t), \gamma (s)) = \gt (\gamma(t) + r \hat{n} (t) - \gamma (s))$ for some exponentially decaying function $\gt \in \mathcal{C}^\infty (\mathbb{R}^2; \mathbb{C}^2)$, 
our assumption \eqref{eq:beta} that $\Gamma$ is approximately a straight line at infinity implies that
\begin{align}\label{eq:partus}
\begin{split}
    \partial_t u_s (\gamma (t) + r\hat{n} (t)) &+ \int_\Gamma \partial_s G_j (\gamma (t) + r \hat{n} (t), \gamma (s)) \mu (s) {\rm d} s\\
    &= \int_\Gamma (\partial_t +\partial_s) G_j (\gamma (t) + r \hat{n} (t), \gamma (s)) \mu (s) {\rm d} s \longrightarrow 0
\end{split}
\end{align}
as $|t| \rightarrow \infty$. Since 
$t \mapsto V(t)$ is smooth and rapidly approaching constant matrices $\cobmat_\pm$ as $t \rightarrow \pm \infty$, the 
function $t \mapsto \cobmat (t) \mmat Q \rho (t) =: \tmu (t)$ is continuously differentiable with
\begin{align*}
    \lim_{t \rightarrow +\infty} (\tmu' - iE \tmu)(t) = \cobmat_+ \mmat \lim_{t \rightarrow +\infty} \{\partial_t (Q \rho (t))-iE Q\rho (t)\}. 
\end{align*}
We now write
\begin{align*}
    \partial_t (Q \rho (t)) = iE Q \rho (t) - 2im^2\int_t^\infty e^{-iE (t-t')} \rho (t') {\rm d} t',
\end{align*}
where the fact that $\rho \in L^1$ implies that the second term on the above right-hand side vanishes as $t \rightarrow +\infty$. 
It follows that
$\lim_{t \rightarrow +\infty} (\tmu' - iE \tmu)(t) =0$, and therefore
\begin{align*}
    &\lim_{t \rightarrow +\infty} (\partial_t - iE) u_s (\gamma (t) + r\hat{n} (t))\\
    &= \lim_{t \rightarrow +\infty} \int_\Gamma (-\partial_s G_j (\gamma (t) + r \hat{n} (t), \gamma (s)) -iEG_j (\gamma (t) + r \hat{n} (t), \gamma (s))) \mu (s) {\rm d} s \\
    &= \lim_{t \rightarrow +\infty} \int_\Gamma (-\partial_s G_j (\gamma (t) + r \hat{n} (t), \gamma (s)) -iEG_j (\gamma (t) + r \hat{n} (t), \gamma (s))) \cobmat (s) \rho (s) {\rm d} s \\
    &\qquad \qquad \qquad +\lim_{t \rightarrow +\infty} \int_\Gamma G_j (\gamma (t) + r \hat{n} (t), \gamma (s)) (\tmu '(s) - iE \tmu (s)) {\rm d} s \ =0,
\end{align*}
where we have used \eqref{eq:partus} to obtain the first equality, and the identity $\mu = \cobmat \cP \rho = \cobmat \rho + \cobmat \mmat Q \rho = \cobmat \rho + \tmu$ (along with integration by parts in $s$) 
to justify the second equality.
A parallel argument shows that
$\lim_{t \rightarrow -\infty} (\partial_t + iE) u_s (\gamma (t) + r\hat{n} (t)) = 0$, and the proof is complete. 
\end{proof}

We next prove 
Theorem \ref{thm:invE}, which is another result that establishes the general well-posedness of our integral equation \eqref{eq:solverho_pfs}. 
To prove this theorem, we will first need to show that 
a related Klein-Gordon equation 
has a unique solution when $E$ is moved into the complex plane.
\begin{lemma}\label{lemma:kgh10}
    Fix $m>0$ and $E \in \mathbb{C}$ such that $0 < |\Re E| < m$ and $\Im E \ne 0$. Define $\omega^2 := m^2 - E^2$, and
    suppose $\psikg \in H^1 (\mathbb{R}^2; \mathbb{C}^2)$ satisfies
    \begin{align}
        (-\Delta + \omega^2) \psikg (x) = 0, \qquad &x \in \Omega_1 \cup \Omega_2 \label{eq:KGPDE}\\
        [[\hat{n} \cdot \nabla \psikg]] (\gamma (t)) = -2im (\hat{n} \cdot \sigma) \sigma_3 \psikg (\gamma(t)), \qquad &t \in \mathbb{R}. \label{eq:KGjump}
    \end{align}
    Then $\psikg \equiv 0$.
\end{lemma}
The proof is similar to the proof of \cite[Lemma 9]{bal2022integral}. We include it here for completeness.
\begin{proof}
    Multiplying both sides of \eqref{eq:KGPDE} by $\psikg^*$ and integrating by parts, we obtain
    \begin{align}\label{eq:Omega2}
        \int_{\Omega_2}( |\nabla u|^2 + \omega^2 |u|^2) {\rm d} x +\int_\Gamma u^* \Big(\frac{\partial u}{\partial n}\Big)_2 d\Sigma = 0,
    \end{align}
    where $\Big(\frac{\partial u}{\partial n}\Big)_2 (\gamma (t)) := \lim_{\Omega_2 \ni y \rightarrow \gamma(t)} \hat{n} (t) \cdot \nabla u(y)$ is the derivative of $u$ with respect to the unit normal vector $\hat{n}$ evaluated in $\Omega_2$, and $\Sigma$ is the surface measure on $\Gamma$.
    A parallel argument shows that
    \begin{align}\label{eq:Omega1}
        \int_{\Omega_1}( |\nabla u|^2 + \omega^2 |u|^2) {\rm d} x -\int_\Gamma u^* \Big(\frac{\partial u}{\partial n}\Big)_1 d\Sigma = 0,
    \end{align}
    with $\Big(\frac{\partial u}{\partial n}\Big)_1 (\gamma (t)) := \lim_{\Omega_1 \ni y \rightarrow \gamma(t)} \hat{n} (t) \cdot \nabla u(y)$ the same normal derivative, but evaluated in $\Omega_1$.
    Combining \eqref{eq:Omega2} and \eqref{eq:Omega1}, we obtain that
    \begin{align*}
        \int_{\mathbb{R}^2}( |\nabla u|^2 + \omega^2 |u|^2) {\rm d} x + \int_\Gamma u^* \Big[\Big[\frac{\partial u}{\partial n}\Big]\Big] d\Sigma=0,
    \end{align*}
    where $[[\frac{\partial u}{\partial n}]]$ denotes the jump in the normal derivative $u$ across $\Gamma$.
    Since $i(\hat{n} \cdot \sigma) \sigma_3 = n_1 \sigma_2 - n_2 \sigma_1$, \eqref{eq:KGjump} implies that
    \begin{align*}
        \Big[\Big[\frac{\partial u}{\partial n}\Big]\Big] (\gamma (t))= 2m (n_2 \sigma_1 - n_1 \sigma_2)u,
    \end{align*}
    hence
    \begin{align}\label{eq:preim}
        \int_{\mathbb{R}^2}( |\nabla u|^2 + \omega^2 |u|^2) {\rm d} x + 2m \int_\Gamma u^* (n_2 \sigma_1 - n_1 \sigma_2) u d\Sigma=0.
    \end{align}
    Now, observe that $\omega^2 = m^2- E^2 = m^2 - (\Re E)^2 + (\Im E)^2 - 2 i \Re E \Im E$ has nonzero imaginary part, by assumption on $E$.
    Since $\sigma_1$ and $\sigma_2$ are Hermitian, all terms in \eqref{eq:preim} not involving $\omega^2$ are real. 
    We conclude that
    \begin{align*}
        \int_{\mathbb{R}^2} |u|^2 {\rm d} x = 0,
    \end{align*}
    and the result is complete.
\end{proof}

We then show that the integral operator defining $u_s$ has a trivial kernel.
\begin{lemma}\label{lemma:mu0}
    Fix $m\ne 0$ and $E \in \mathbb{C}$ such that $0 < |\Re E| < |m|$ and $\Im E \ne 0$.
    Set $\omega \in \mathbb{C}$ such that $\omega^2 = m^2 - E^2$ and $\Re \omega > 0$.
    Define
    \begin{align}\label{eq:usproof}
         u_s (x) = \begin{cases}
        \int_\Gamma G_2 (x, \gamma (t)) \mu (t) {\rm d} t, & x \in \Omega_2\\
        \int_\Gamma G_1 (x, \gamma (t)) \mu (t) {\rm d} t, & x \in \Omega_1,
        \end{cases},
    \end{align}
    where
    \begin{align*}
    G_j (x,y) := \frac{1}{2\pi} (-i \partial_{x_1} \sigma_1 -i \partial_{x_2} \sigma_2 + (-1)^j m \sigma_3 + E) K_0 (\omega |x-y|)
    \end{align*}
    for $j=1,2$.
    If $u_s \equiv 0$ and $\mu \in L^2_\alpha \cap C^\infty (\mathbb{R}; \mathbb{C}^2)$ for some $\alpha > 0$, then $\mu \equiv 0$.
\end{lemma}
\begin{proof}
    Suppose $u_s \equiv 0$ and $\mu \in L^2_\alpha \cap C^\infty (\mathbb{R}; \mathbb{C}^2)$. 
    It follows from the regularity of $\mu$ and the derivation of the integral equation in Section \ref{subsec:naive} that
    \begin{align*}
        [[u_s]](\gamma (t))=
        2m \sigma_3 \cS_\omega [\mu](t) + i (\hat{n}(t) \cdot \sigma) \mu (t),
    \end{align*}
    hence
    \begin{align*}
        \mu (t) = 2m (-i\hat{n} (t) \cdot \sigma)\sigma_3 \cS_\omega [\mu](t).
    \end{align*}
    By the well-known property \eqref{eq:jumpgradS} of the single layer potential, it follows that $$\psikg (x) := \frac{1}{2\pi}\int_{\mathbb{R}} K_0 (\omega |x-\gamma (t')|) \mu (t') {\rm d} t'$$ satisfies \eqref{eq:KGjump}.
    Since $G_\omega (x,y) =\frac{1}{2\pi} K_0(\omega |x-y|)$ is the Green's function for the Klein-Gordon operator $-\Delta + \omega^2$, we know that $\psikg$ satisfies \eqref{eq:KGPDE} as well.
    If $\mu \in L^2_\alpha (\mathbb{R}; \mathbb{C}^2)$, then $\psikg \in H^1 (\mathbb{R}^2; \mathbb{C}^2)$, meaning that $\psikg$ satisfies all the assumptions of Lemma \ref{lemma:kgh10} and thus $\psikg \equiv 0$. Again appealing to \eqref{eq:jumpgradS}, we conclude that
    \begin{align*}
        0 = [[\hat{n} \cdot \nabla u]](\gamma (t))=-\mu (t)
    \end{align*}
    for all $t \in \mathbb{R}$, and the result is complete.
\end{proof}

We are now ready to prove 
that the integral equation \eqref{eq:solverho_pfs} admits a unique solution for almost all choices of $E$.
\begin{proof}[Proof of Theorem \ref{thm:invE}]
    Fix $\alpha > 0$ sufficiently small, and let $E \in \mathbb{C}$ such that
    $\Re E \in [-|m|+\eps, -\eps] \cup [\eps, |m|-\eps]$ and $\Im E > 0$. 
    We write \eqref{eq:solverho} as
    \begin{align*}
        (\flatL \cP + (\cL- \flatL)\cP)[\rho] (t) =i \cobmat^* \hat{n} \cdot \sigma [[u_i]] (\gamma (t)),
    \end{align*}
    which is equivalent to
    \begin{align}\label{eq:cMrho}
        (1+\cM) [\rho] (t) = i(\flatL \cP)^{-1}\cobmat^* \hat{n} \cdot \sigma [[u_i]] (\gamma (t)),
    \end{align}
    where $\cM := (\flatL \cP)^{-1} (\cL- \flatL)\cP$. We have shown with Lemma \ref{lemma:MHS} that $\cM$ is compact on $L^2_\alpha (\mathbb{R}; \mathbb{C}^2)$, therefore
    \eqref{eq:cMrho} has a unique solution if and only if $\cM$ does not have an eigenvalue of $-1$.

    Set \begin{align*}
    \cobmat := \frac{1}{\sqrt{2}}\begin{pmatrix}
        1 & in_1 + n_2\\
        -in_1 + n_2 & -1
    \end{pmatrix},
    \end{align*}
    and define $\mu := \cP \cobmat \rho$, where $\cP$ is defined by
    \begin{align*}
    \cP [\eta] (t) &= \eta (t) + \mmat \frac{m^2}E \int_{\Rm} e^{iE|t-t'|} \eta (t'){\rm d}t',
    \end{align*}
    with
    $\mmat$ as defined in~\eqref{eqn:mdef}.

    Suppose $(1+\cM) [\rho] (t) = 0$ for some $\rho \in L^2_\alpha (\mathbb{R}; \mathbb{C}^2)$.
    This means $\rho$ satisfies \eqref{eq:solverho} with $u_i \equiv 0$. Therefore, taking $f_1 = f_2 \equiv 0$ in Theorem \ref{thm:usol}, we obtain that $u_s$ defined by \eqref{eq:usproof} satisfies
    \begin{align}\label{eq:PDEproof}
    \begin{split}
        (-i \partial_{x_1} \sigma_1- i \partial_{x_2} \sigma_2 + m \sigma_3 - E) u_s (x) = 0, \qquad &x \in \Omega_2\\
        (-i \partial_{x_1} \sigma_1- i \partial_{x_2} \sigma_2 - m \sigma_3 - E) u_s (x) = 0, \qquad &x \in \Omega_1\\
        \lim_{y\to x} u_s(y) = \lim_{z \to x} u_s(z), \qquad & y \in \Omega_2, z \in \Omega_1, x \in \Gamma
    \end{split}
    \end{align}
    Since $\Im E > 0$, we know that $\cP$ is a bounded operator on $L^2_\alpha (\mathbb{R}; \mathbb{C}^2)$, and therefore $\mu \in L^2_\alpha (\mathbb{R}; \mathbb{C}^2)$.
    As in the proof of Lemma \ref{lemma:mu0}, this implies that $u_s \in H^1 (\mathbb{R}^2; \mathbb{C}^2)$.
    Applying the operator $-i \partial_{x_1} \sigma_1- i \partial_{x_2} \sigma_2 + m \sigma_3 + E$ to 
    the first line of \eqref{eq:PDEproof} and $-i \partial_{x_1} \sigma_1- i \partial_{x_2} \sigma_2 + m \sigma_3 + E$ to the second, it follows that 
    \begin{align*}
        (-\Delta + \omega^2) u_s (x) = 0, \qquad &x \in \Omega_1 \cup \Omega_2\\
        [[\hat{n} \cdot \nabla u_s]] (\gamma (t)) = -2im (\hat{n} \cdot \sigma) \sigma_3 u_s (\gamma(t)), \qquad &t \in \mathbb{R}.
    \end{align*}
    Therefore, $u_s \equiv 0$ by Lemma \ref{lemma:kgh10}. We have thus shown that $u_s$ and $\mu$ satisfy the assumptions of Lemma \ref{lemma:mu0}, meaning that $\mu \equiv 0$. 
    In Section \ref{subsec:flat}, we showed that
    \begin{align*}
        \mathcal{F}_{r \rightarrow \xi} \{\cP [\eta] (t)\} = 1 - \mmat \frac{2i m^2}{\xi^2-E^2},
    \end{align*}
    which combined with the definition $\mu = \cP \cobmat \rho$ proves that $\cobmat \rho \equiv 0$. The fact that $\rho \equiv 0$ then follows from the invertibility of $\cobmat$.

    We have now shown that if $(1+\cM) [\rho] (t) = 0$ for some $\rho \in L^2_\alpha (\mathbb{R}; \mathbb{C}^2)$, then $\rho \equiv 0$. This means that $\cM$ does not have an eigenvalue of $-1$ when $\Re E \in [-|m|+\eps, -\eps] \cup [\eps, |m|-\eps]$ and $\Im E > 0$. 
    Since all integral operators here are holomorphic in $E \in [\eps, |m|-\eps] \times (-\alpha/2, \infty) \subset \mathbb{C}$ and $E \in [-|m|+\eps, -\eps] \times (-\alpha/2, \infty)$, it follows from analytic perturbation theory \cite[Theorem VII.1.9]{kato2013perturbation} that $\cM$ has an eigenvalue of $-1$ for at most a finite number of values of $E \in [-|m|+\eps, -\eps] \cup [\eps, |m|-\eps] \subset \mathbb{R}$. Therefore,
    the integral equations \eqref{eq:cMrho} and hence \eqref{eq:solverho} have a unique solution $\rho \in L^2_\alpha (\mathbb{R}; \mathbb{C}^2)$ for all $E$ except those exceptional values, and the result is complete.
\end{proof}

We now move on to the two-interface problem, proving that it, too, is well posed for almost all choices of parameters.

\begin{proof}[Proof of Theorem \ref{thm:invL2}]
    By 
    Lemmas \ref{lemma:invLP} and \ref{lemma:MHS}, the operators $\cL_j \cP_j$ have bounded inverse (on $L^2_\alpha$) whenever $\lambda$ is sufficiently large, with $(\cL_j \cP_j)^{-1} = (1 + \cM_j)^{-1} (\cL_{j,0} \cP_j)^{-1}$ bounded uniformly in $\lambda$. 
    Here $\cL_{j,0}$ is the operator $\cL_j$ in the case of a flat interface, meaning that
    $\cL_{j,0} := 1+(-1)^j 2m\sigma_3 \cSflat$ with
    \begin{align*}
        \cSflat [\mu] (t) := \frac{1}{2\pi} \int_{\mathbb{R}} K_0 (\omega |t-t'|) \mu (t') {\rm d} t'
    \end{align*}
    as before,
    and $\cM_j := (\cL_{j,0} \cP_j)^{-1}(\cL_j -\cL_{j,0})\cP_j.$
    This means \eqref{eq:solverho2} reduces to
    \begin{align}\label{eq:red}
        \Bigg( 1 + 
        \begin{pmatrix}
            0 & (\cL_1 \cP_1)^{-1} \cK_2 \cP_2\\
            (\cL_2 \cP_2)^{-1} \cK_1 \cP_1 & 0
        \end{pmatrix}
        \Bigg)
        \begin{pmatrix}
            \rho_1\\
            \rho_2
        \end{pmatrix}
        =
        \begin{pmatrix}
        (\cL_1 \cP_1)^{-1} \cobmat_1^*(i \hat{n}_1 \cdot \sigma) [[u_i]]_1\\
        (\cL_2 \cP_2)^{-1} \cobmat_2^*(i \hat{n}_2 \cdot \sigma) [[u_i]]_2
    \end{pmatrix}.
    \end{align}
    Since $\Gamma_1$ and $\Gamma_2$ do not intersect, the assumption \eqref{eq:gammainfty2} implies that 
    $|G_1 (\gamma_1 (t), \gamma_2 (t'))| \le C e^{-\dc \lambda (|t| + |t'|)}$ uniformly in $t \in \Gamma_1$ and $t' \in \Gamma_2$, for some positive constants $C$ and $\dc$.
    Therefore, the operators $(\cL_i \cP_i)^{-1} \cK_j \cP_j$ 
    are bounded on $L^2_\alpha$ with norm that decays exponentially in $\lambda$. We conclude that when $\lambda$ is sufficiently large, \eqref{eq:red} admits a unique solution $(\rho_1, \rho_2) \in L^2_\alpha$, and thus so does \eqref{eq:solverho2}. Since all operators are holomorphic in $\lambda \in [1, \infty) \times (-\delta, \delta)$ for $\delta > 0$ sufficiently small, the result follows from Kato perturbation theory \cite[Theorem VII.1.9]{kato2013perturbation} as in Theorem \ref{thm:invL}.
\end{proof}

Next, we prove that the solution of our two-interface boundary integral equation \eqref{eq:solverho2} generates a solution of the corresponding Dirac equation (\ref{eq:PDE2}, \ref{eq:out2}).
\begin{proof}[Proof of Theorem \ref{thm:usol2}]
    This argument is almost identical to the proof of Theorem \ref{thm:usol} above.
    As before, the first three lines of \eqref{eq:PDE2} and the second line of \eqref{eq:out2} follow from dominated Lebesgue, while the continuity conditions (last two lines of \eqref{eq:PDE2}) are an immediate consequence of the derivation of our integral equation \eqref{eq:solverho2} in Section \ref{sec:two} and regularity of $(\rho_1, \rho_2)$.
    To verify the radiation conditions specified by the first line of \eqref{eq:out2}, observe that assumption \eqref{eq:gammainfty2} and the rapid decay of $G_1$ at infinity imply that $u_s$ (defined by \eqref{eq:us2}) and its derivatives are well approximated by the function
    \begin{align*}
    u_{s,\infty,2}(x) :=
        \begin{cases}
        \int_\mathbb{R} G_2 (x, \gamma_2 (t')) \mu_2 (t') {\rm d} t', & x \in \Omega_2\\
        \int_\mathbb{R} G_1 (x, \gamma_2 (t')) \mu_2 (t') {\rm d} t', & x \in \Omega_1
        \end{cases}
    \end{align*}
    when $x = \gamma_2 (t) + r \hat{n}_2(t)$ and $|t|$ large, while
    \begin{align*}
    u_{s,\infty,1}(x) :=
        \begin{cases}
        \int_\mathbb{R} G_1 (x, \gamma_1 (t')) \mu_1 (t') {\rm d} t', & x \in \Omega_1\\
        \int_\mathbb{R} G_2 (x, \gamma_1 (t')) \mu_1 (t') {\rm d} t', & x \in \Omega_0
        \end{cases}
    \end{align*}
    approximates $u_s$ when $|t|$ is large and
    $x = \gamma_1 (t) + r \hat{n}_1(t)$.
    Since the $u_{s,\infty,j}$ have the form of their one-interface analogues \eqref{eq:us}, the proof of the first line of \eqref{eq:out0} in Theorem \ref{thm:usol} is easily applied to 
    verify the first line of \eqref{eq:out2}. This completes the result.
\end{proof}

We will conclude this section with the proof of Theorem \ref{thm:invE2},
which uses the following two-interface analogue of Lemma \ref{lemma:mu0}.
\begin{lemma}\label{lemma:mu0_2}
    Fix $m\ne 0$ and $E = E_0 + i\lambda$, where the real numbers $E_0$ and $\lambda$ satisfy $0 < |E_0| < |m|$ and $\lambda > 0$.
    Set $\omega \in \mathbb{C}$ such that $\omega^2 = m^2 - E^2$ and $\Re \omega > 0$.
    Define
    \begin{align}\label{eq:usproof2}
         u_s (x) = \begin{cases}
        \int_\mathbb{R} G_2 (x, \gamma_2 (t)) \mu_2 (t) {\rm d} t, & x \in \Omega_2\\
        \int_\mathbb{R} G_1 (x, \gamma_2 (t)) \mu_2 (t) {\rm d} t
        + \int_\mathbb{R} G_1 (x, \gamma_1 (t)) \mu_1 (t) {\rm d} t, & x \in \Omega_1\\
        \int_\mathbb{R} G_2 (x, \gamma_1 (t)) \mu_1 (t) {\rm d} t, & x \in \Omega_0
    \end{cases},
    \end{align}
    where
    \begin{align*}
    G_j (x,y) := \frac{1}{2\pi} (-i \partial_{x_1} \sigma_1 -i \partial_{x_2} \sigma_2 + (-1)^j m \sigma_3 + E) K_0 (\omega |x-y|)
    \end{align*}
    for $j=1,2$. 
    Then for any $\lambda$ sufficiently large, the unique function $(\mu_1, \mu_2) \in L^2 \cap C^\infty (\mathbb{R}; \mathbb{C}^4)$ for which $u_s \equiv 0$ is $(\mu_1, \mu_2) \equiv 0$.
\end{lemma}
\begin{proof}
    Suppose $u_s \equiv 0$ and $(\mu_1, \mu_2) \in L^2 \cap C^\infty (\mathbb{R}; \mathbb{C}^4)$. Using that $$\lim_{\Omega_1 \ni x \to \gamma_2 (t)}u_s (x) = \lim_{\Omega_2 \ni x \to \gamma_2 (t)} u_s (x), \qquad \lim_{\Omega_0 \ni x \to \gamma_1 (t)}u_s (x) = \lim_{\Omega_1 \ni x \to \gamma_1 (t)} u_s (x),$$ it follows from the regularity of $(\mu_1,\mu_2)$ that
    \begin{align*}
        (i \hat{n}_2 (t) \cdot \sigma) \mu_2 (t) + 2m \sigma_3 \cS_\omega [\mu_2] (t) - \int_\mathbb{R} G_1 (\gamma_2 (t), \gamma_1 (t')) \mu_1 (t') {\rm d} t' &= 0,\\
        (i \hat{n}_1 (t) \cdot \sigma) \mu_1 (t) - 2m \sigma_3 \cS_\omega [\mu_1] (t) + \int_\mathbb{R} G_1 (\gamma_1 (t), \gamma_2 (t')) \mu_2 (t') {\rm d} t' &= 0.
    \end{align*}
    Multiplying the top (resp. bottom) line by $-i \hat{n}_2 (t) \cdot \sigma$ (resp. $-i \hat{n}_1 (t) \cdot \sigma$), we obtain
    \begin{align*}
    (I + \opn)
    \begin{pmatrix}
                \mu_1 \\
                \mu_2
            \end{pmatrix}
            = \begin{pmatrix}
                0 \\
                0
            \end{pmatrix},
            \qquad 
            \opn := \begin{pmatrix}
            (i \hat{n}_1 \cdot \sigma) 2m \sigma_3 \cS_\omega & -(i \hat{n}_1 \cdot \sigma) \bar{\cK}_2\\
            (i \hat{n}_1 \cdot \sigma) \bar{\cK}_1 & -(i \hat{n}_2 \cdot \sigma) 2m \sigma_3 \cS_\omega
        \end{pmatrix}
    \end{align*}
    where the $\bar{\cK}_j$ are defined by \eqref{eq:bar_cK}. Thus it suffices to show that the operator $\opn : L^2 \to L^2$ has norm strictly less than $1$ whenever $\lambda>0$ is sufficiently large.

    It follows from the definition $\omega^2 = m^2 - E_0^2 + \lambda^2 - 2i\lambda E_0$ and $\Re \omega > 0$ that $\Re \omega > \lambda$. Thus,
    each $2\times 2$ block of $\opn$ is an integral operator with a kernel $k(s,t)$ bounded by $C \log (\lambda^{-1} |t-s|^{-1})$ whenever $\lambda |t-s| < 1/2$ and $C e^{-\eta \lambda |t-s|}$ whenever $\lambda |t-s| \ge 1/2$, where $C$ and $\eta$ are fixed positive constants that are independent of $\lambda$. 
    We can therefore use the identity $\norm{g_1 * g_2}_{L^2} \le \norm{g_1}_{L^1}\norm{g_2}_{L^2}$ with $g_1 (t) = (\left|\log(\lambda|t|)\right|+1) e^{-\eta\lambda |t|}$ and $g_2 = \mu$ to show that $\norm{\int_\mathbb{R} k(\cdot, s) \mu(s) {\rm d} s}_{L^2} \le C \lambda^{-1} \norm{\mu}_{L^2}$, where $C>0$ is independent of $\lambda$. 
    We conclude that $\norm{\opn}_{L^2 \to L^2} \le C \lambda^{-1}$ and the result is complete.
\end{proof}

We are now ready to prove our second result that establishes the general invertibility of the two-interface integral equation \eqref{eq:solverho2}.
\begin{proof}[Proof of Theorem \ref{thm:invE2}]
    The strategy is the same as for Theorem \ref{thm:invE}.
    Let $\alpha > 0$ be sufficiently small.
    By Kato holomorphic perturbation theory \cite{kato2013perturbation}, it suffices to show that the integral equation
    \begin{align}\label{eq:solverho2_appendix}
    \begin{pmatrix}
        \cL_1 & \cK_2\\
        \cK_1 & \cL_2
    \end{pmatrix}
    \begin{pmatrix}
        \cP_1 & 0\\
        0 & \cP_2
    \end{pmatrix}
    \begin{pmatrix}
        \rho_1\\
        \rho_2
    \end{pmatrix} =
    \begin{pmatrix}
        0\\
        0
    \end{pmatrix},
    \qquad
    \begin{pmatrix}
        \rho_1\\
        \rho_2
    \end{pmatrix} \in L^2_\alpha (\mathbb{R}; \mathbb{C}^4)
    \end{align}
    admits only the trivial solution $(\rho_1, \rho_2) \equiv 0$ whenever 
    $E = m/2 + i\lambda$ with $\lambda > 0$ sufficiently large.
    To this end, assume that $(\rho_1, \rho_2)$ satisfies \eqref{eq:solverho2_appendix} and set $\mu_j := \cobmat_j \cP_j \rho_j$. Then by Theorem \ref{thm:usol2},
    \begin{align*}
         u_s (x) = \begin{cases}
        \int_\mathbb{R} G_2 (x, \gamma_2 (t)) \mu_2 (t) {\rm d} t, & x \in \Omega_2\\
        \int_\mathbb{R} G_1 (x, \gamma_2 (t)) \mu_2 (t) {\rm d} t
        + \int_\mathbb{R} G_1 (x, \gamma_1 (t)) \mu_1 (t) {\rm d} t, & x \in \Omega_1\\
        \int_\mathbb{R} G_2 (x, \gamma_1 (t)) \mu_1 (t) {\rm d} t, & x \in \Omega_0
    \end{cases}
    \end{align*}
    satisfies
    \begin{align}
    \begin{split}
        (-i \partial_{x_1} \sigma_1- i \partial_{x_2} \sigma_2 + m \sigma_3 - E) u_s (x) = 0, \qquad &x \in \Omega_0 \cup \Omega_2\\
        (-i \partial_{x_1} \sigma_1- i \partial_{x_2} \sigma_2 - m \sigma_3 - E) u_s (x) = 0, \qquad &x \in \Omega_1\\
        \lim_{y\to x} u_s(y) = \lim_{z \to x} u_s(z), \qquad & y \in \Omega_2, z \in \Omega_1, x \in \Gamma_2\\
        \lim_{y\to x} u_s(y) = \lim_{z \to x} u_s(z), \qquad & y \in \Omega_1, z \in \Omega_0, x \in \Gamma_1
    \end{split}
    \end{align}
    Moreover, since 
    $\lambda > 0$ and $(\rho_1, \rho_2) \in L^2_\alpha$, it follows that $(\mu_1, \mu_2) \in L^2_{\eta}$ for some $\eta>0$. Therefore, $u_s \in H^1 (\mathbb{R}^2;\mathbb{C}^2)$ by the decay of the $G_j$.
    As in the proof of Lemma \ref{lemma:kgh10}, a standard
    integration by parts argument then
    implies that $u_s \equiv 0$.
    We then conclude from Lemma \ref{lemma:mu0_2} that $(\mu_1, \mu_2) \equiv 0$.
    Following the proof of Theorem \ref{thm:invE}, the Fourier representation of the $\cP_j$ and invertibility of the $\cobmat_j$ imply that $(\rho_1, \rho_2) \equiv 0$. This completes the result.
\end{proof}

\subsection*{Acknowledgments} 
G. Bal's work was supported in part by NSF grants DMS-1908736 and EFMA-1641100.

 \bibliographystyle{elsarticle-num} 
 \bibliography{refs}





\end{document}